\documentclass[11pt,reqno]{amsart}

\usepackage{amsmath,amsfonts,amsthm,amssymb,amsxtra}
\usepackage[colorlinks,citecolor=red,hypertexnames=false]{hyperref} 
\usepackage{color}
\usepackage{tikz}
\usepackage[shortlabels]{enumitem}
\usepackage{bbm} 
\usepackage{stmaryrd}
\usepackage{nicematrix}
\usepackage{mathrsfs} 
\usepackage{float}
\usepackage{graphicx,subfigure}

\usepackage[margin=1.1in]{geometry}

\newtheorem{theorem}{Theorem}[section]

\newtheorem{proposition}[theorem]{Proposition}
\newtheorem{lemma}[theorem]{Lemma}

\newtheorem{claim}{Claim}

\theoremstyle{definition}

\theoremstyle{remark}

\newtheorem{remark}{Remark}[section]

\numberwithin{equation}{section}
\allowdisplaybreaks

\newcommand{\C}{\mathbb{C}}

\renewcommand{\epsilon}{\varepsilon}

\newcommand{\N}{\mathbb{N}}

\newcommand{\G}{\mathbb{G}}

\renewcommand{\phi}{\varphi}
\newcommand{\R}{\mathbb{R}}

\newcommand{\T}{\mathbb{T}}

\newcommand{\Z}{\mathbb{Z}}

\newcommand{\E}{\mathbb{E}}

\newcommand{\one}{\mathbbm{1}}

\newcommand{\wt}{\widetilde}

\DeclareFontFamily{U}{mathx}{}
\DeclareFontShape{U}{mathx}{m}{n}{<-> mathx10}{}
\DeclareSymbolFont{mathx}{U}{mathx}{m}{n}
\DeclareMathAccent{\widehat}{0}{mathx}{"70}
\DeclareMathAccent{\widecheck}{0}{mathx}{"71}

\newcommand{\ipc}[2]{ \langle #1 , #2  \rangle }

\renewcommand{\P}{\mathbb{P}}

\newcommand{\eps}{\varepsilon}

\newcommand{\vertiii}[1]{{\left\vert\kern-0.25ex\left\vert\kern-0.25ex\left\vert #1
    \right\vert\kern-0.25ex\right\vert\kern-0.25ex\right\vert}}

\DeclareMathOperator{\supp}{supp}

\newcommand\numberthis{\addtocounter{equation}{1}\tag{\theequation}}

\begin{document}

 \title[Spectrum and Lifshitz tails for SG]{Spectrum and Lifshitz tails for the Anderson model on the Sierpinski gasket graph}

\author[L. Shou, W. Wang, S. Zhang]{Laura Shou, Wei Wang, Shiwen Zhang}

\begin{abstract}
 In this work, we study the Anderson model on the Sierpinski gasket graph. We first identify the almost sure spectrum of the Anderson model when the support of the random potential has no gaps. We then prove the existence of the integrated density states of the Anderson model and show that it has Lifshitz tails with Lifshitz exponent determined by the ratio of the volume growth rate and the random walk dimension of the Sierpinski gasket graph. 
\end{abstract}

  \maketitle

\section{Introduction}
Random Schr\"odinger operators play an important role in describing the quantum state of a particle in a disordered material.
The tight-binding approximation simplifies the model further by restricting the positions of electrons, leading to a Hamiltonian given by a discrete random Schr\"odinger $-\Delta+V$ acting on a Hilbert space $\ell^2(\mathbb G)$, with $\mathbb G$ the vertices of a graph, $\Delta$ the graph Laplacian, and $V$ a random potential, acting as a multiplication operator $V(x)=V_\omega(x)$.  The simplest, but also one of the most prominent, form of the random potential is to take $\{V(x)\}_{x\in\G}$ as independent and identically distributed (i.i.d.) random variables. The corresponding random operator is called the \emph{Anderson model}. 
The Anderson model, along with more general random operators, have been long investigated in the ergodic setting, on a vertex-transitive graph such as the standard $\Z^d$ lattice, or the Bethe lattice (a regular tree graph); see a thorough discussion in e.g. \cite{aizenman2015random}.  Fractal graphs are notable examples that are not vertex-transitive. 
 
 In this paper, we study spectral properties of the  Anderson model on such a fractal graph, the Sierpinski gasket graph, also called the \emph{Sierpinski lattice}\footnote{The `true' Sierpinski gasket $K$ is a compact fractal subset of $\R^2$, constructed by a self-similar iterated
function scheme, see e.g. \cite{fal1990,barlow2017random}. The Sierpinski lattice $\G$, as an finite combinatorial graph, is defined with a similar structure, whose large-scale
structure mimics the microscopic structure of the Sierpinski gasket $K$. 
}. There is a large literature about the free Laplacian $\Delta$ on the Sierpinski gasket or Sierpinski lattice, as well as more general nested fractals and p.c.f. self-similar sets or graphs, see e.g. \cite{alex1984,ram1984, kusuoka1987diffusion,goldstein1987random, barlow1988, kig1989,  
fuku1988,shima1991eigenvalue,fuku1992,lind1990,lapidus1994,malo1993,
teplyaev1998spectral,dalrymple1999fractal,stri1999}, which is far from a complete list.  For the Anderson model and other random Hamiltonians on fractals, there are many works in the physics literature, see e.g. \cite{schreib1996,stic2016,kosi2017,manna2024}, among others.  But limited work has been done for the Anderson model on fractals in the math literature, except in, e.g. \cite{shima1991lifschitz, pietruska1991lifschitz,ma95,balsam2023density}.   Lacking the ergodic setting, it is not easy in general to determine the full spectrum of a random operator, nor its different spectral components. In particular, to the best to our knowledge, there is no work on the (topological) structure or identification of deterministic spectra for the Anderson model on the Sierpinski gasket or Sierpinski lattice, nor on other fractal sets or graphs.  

The first result of this paper is the almost-sure spectrum of the discrete random Schr\"odinger operator on  the Sierpinski lattice when the support of the random potential has no gaps. There are also partial results of the spectrum for general random potentials. On the other hand, the works of \cite{shima1991lifschitz, pietruska1991lifschitz,balsam2023density} concern   the continuous/differential random Schr\"odinger operator  on the Sierpinski gasket (rather than the discrete/graphical Sierpinski lattice), and prove the existence and the Lifshitz-type singularity of the integrated density of states (IDS) for the random Schr\"odinger operator. 
There are no discrete analogue of these results. The second part of this paper obtains the existence of the IDS for the Anderson model on the Sierpinski lattice, and establishes the Lifshitz tail estimates of the IDS near the bottom of the spectrum.

 \subsection{Main results}
 
 To state our main results, we first introduce the Sierpinski lattice, and the associated Anderson model. 
  Let $T_0\subseteq \R^2$ be the unit equilateral triangle with the vertex set  $\G_0=\{(0,0),(1,0),(1/2,\sqrt 3/2)\}=\{ a_1,a_2,a_3\}$.  
Now, recursively define $T_n$ by
\begin{align}\label{eqn:T-rec}
  T_{n+1}=T_n\cup \big(T_n+2^na_2\big) \cup \big(T_n+2^na_3\big),
  \ n\ge 0. 
\end{align}
Each $T_n$ consists of $3^n$ many translations of the unit triangle $T_0$. 
Let $\G_n$ be the collection of all vertices of the triangles in $T_n$. The set of edges $\mathcal E_n=\{(x,y):x,y\in \G_n\}$ is defined by the relation $(x,y)\in \mathcal E_n$ iff there exists a triangle $T$ on side 1 in $T_n$ with $x,y\in T$.

Let
\begin{align}\label{eqn:SG-def}
    \G=\bigcup_{n\ge 0}(\G_n\, \cup \,  \G_n'), \ \ \mathcal E=\bigcup_{n\ge 0}(\mathcal E_n\, \cup \mathcal E_n'), \ \ \Gamma=(\G,\mathcal E), 
\end{align}
where $\G_n'$ and $\mathcal E_n'$ are the symmetric image of $\G_n$ and $\mathcal E_n$ with respect to the $y$-axis respectively; see Figure~\ref{fig:V3}. 
\begin{figure}
    \centering
    \includegraphics[width=.8\linewidth]{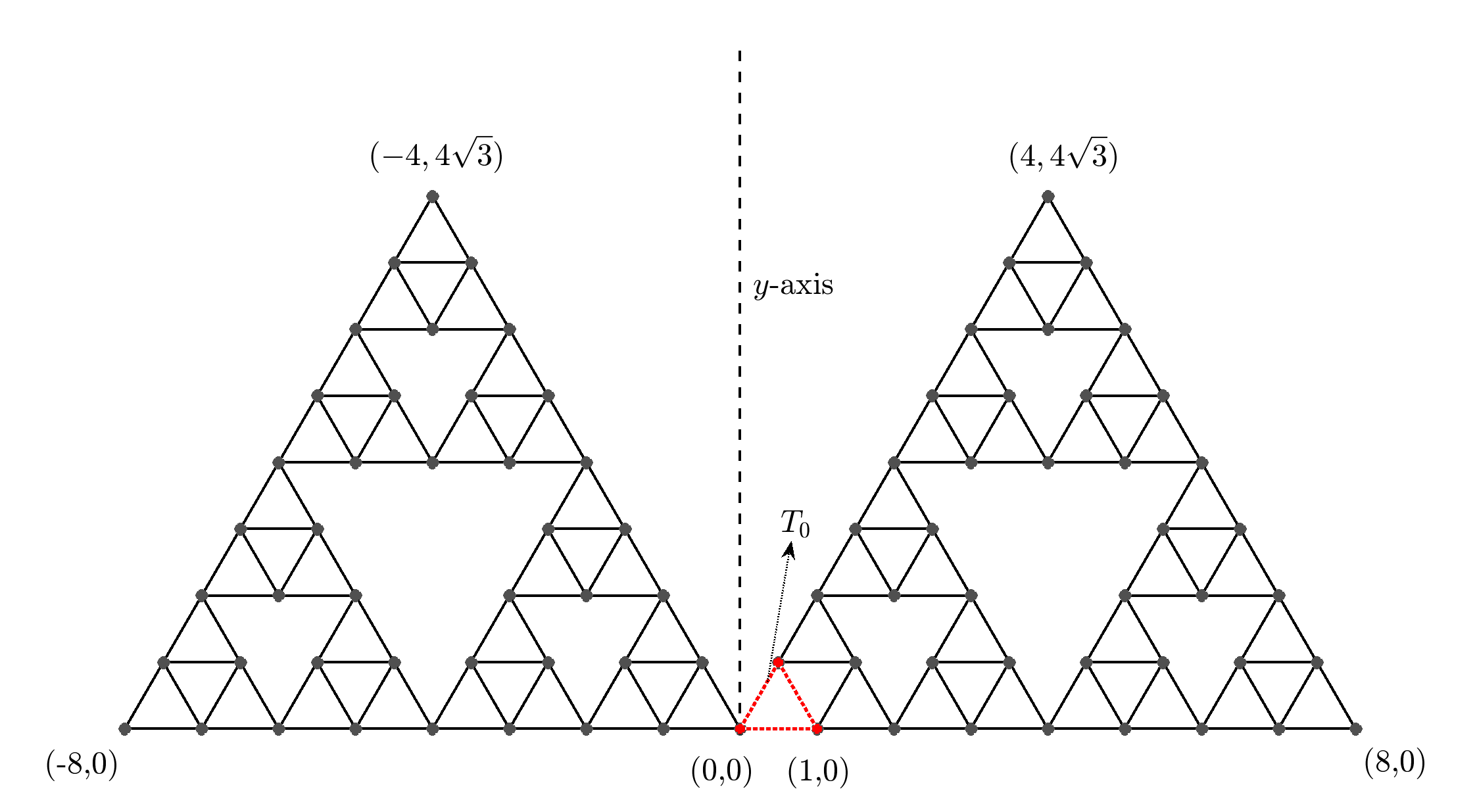}
\caption{
The unit triangle $T_0$ is located next to the origin, with vertices $\G_0=\{(0,0),(1,0),(1/2,\sqrt3/2)\}$.
The right hand side of the $y$-axis is the 3rd step of construction $T_3$ and the left hand side is its reflective mirror with respect to the $y$-axis. The picture contains $2\times 27$ many unit triangles $T$, which are all translations of $T_0$.  The dots form the vertex set $\G_3'\cup\G_3$. The edge set $\mathcal E_3'\cup\mathcal E_3$ consists of all edges of length 1 of the unit triangles. }
    \label{fig:V3}
\end{figure}
 $\Gamma=(\G,\mathcal E)$ is called the (infinite) Sierpinski gasket graph, or Sierpinski lattice, with the vertex set $\G$ and the edge set $\mathcal E$. We write $x\sim y$ to mean $(x,y)\in \mathcal E$, and say that $y$ is a neighbor
of $x$.

The  (combinatorial) graph Laplacian on the Sierpinski lattice is given by
\begin{align}\label{eqn:Lap}
   \Delta f(x)= \sum_{y\in \G:y\sim x}\left(f(y)- f(x)\right), \ x\in \G,  
\end{align}
acting on $\ell^2(\G)$ equipped with the usual inner product $\langle f,g\rangle=\sum_{x\in\G}  f(x)\overline{ g(x)}$. 
 The Anderson model on the Sierpinski lattice is given by the random Hamiltonian $H_\omega=-\Delta+V_\omega$ on $\ell^2(\G)$:
\begin{align}\label{eqn:AM}
    H_\omega f(x)=-\sum_{y\in \G: y\sim x}(f(y)-f(x))+V_\omega(x)f(x), \ \ x\in \G,
\end{align}
where $V_\omega$ is a random potential, acting as the usual multiplicative operator. $\{V_\omega(x)\}_{x\in \G}$ are independent, identically distributed (i.i.d.) random variables, with a common distribution $P_0$.  
We denote by $ \supp P_0$ the (essential) support of the measure $P_0$, defined as 
 \begin{align}\label{eqn:Vsupp}
    \supp P_0=\{x\in \R\ | \ P_0(x-\varepsilon,x+\varepsilon)>0\ {\rm for\ all\ } \varepsilon>0\ \}.
\end{align}

The first result of the paper is the a.s. structure of the spectrum of $H_\omega=-\Delta+V_\omega$. 
\begin{theorem}\label{thm:AM-as-spectrum}
 Let $H_\omega=-\Delta+V_\omega$ be the Anderson model \eqref{eqn:AM} on the Sierpinski lattice $\G$. Then almost surely,  
    \begin{align}\label{eqn:AM-as-spectrum1}
 \sigma(-\Delta)+ \supp P_0      \subseteq  \sigma(H_\omega)\subseteq \sigma(-\Delta)+[\inf V_\omega,\sup V_\omega]     , 
    \end{align}
    and
    \begin{align}\label{eqn:AM-as-spectrum3}
   \sigma(H_\omega)\subseteq[0,6]+ \supp P_0 
    \end{align}
    In particular, if the essential support of the potential is an interval, i.e., $ \supp P_0 =[a,b]$ for $-\infty\le a<b\le \infty$, then the spectrum of $H_\omega$ is almost surely a constant set given by 
    \begin{align}\label{eqn:AM-as-spectrum2}
      \sigma(H_\omega)    =\sigma(-\Delta)+ \supp P_0   .
    \end{align}
\end{theorem}

\begin{remark}
The spectrum of the free Laplacian $-\Delta$ on the Sierpinski lattice as a set was first computed by \cite{belli1982stab,belli1988ren,fuku1992}. The nature of $\sigma(-\Delta)$ and the structure of eigenfunctions were later determined in \cite{teplyaev1998spectral}. As shown in \cite[Theorem 2]{teplyaev1998spectral}, the spectrum $\sigma(-\Delta)\subseteq[0,6]$ consists of isolated eigenvalues and a Cantor set. In particular,  $0=\inf \sigma(-\Delta)$ is in the essential/Weyl spectrum, and $6=\sup \sigma(-\Delta)$ is the largest isolated eigenvalue of infinite multiplicity.

For the Anderson model $-\Delta+V_\omega$, from Theorem~\ref{thm:AM-as-spectrum} we only know that the almost sure constant spectrum structure \eqref{eqn:AM-as-spectrum2} holds when there are no gaps in the support of the potential, e.g. for the uniform distribution, Gaussian distribution, etc. 
For general distributions, we do not know how to show $ \sigma(-\Delta+V_\omega)$ is a non-random set. 
In addition, the relation \eqref{eqn:AM-as-spectrum2} may not hold for, e.g. Bernoulli distributions; see for example numerical experiments in Figure~\ref{fig:gasket-ids}.
Given a Bernoulli potential with 
 $ \supp P_0 =\{a,b\}$, 
 the best that we obtain by \eqref{eqn:AM-as-spectrum3} is
 \begin{align}
     \sigma(-\Delta+V_\omega)\subseteq[0,6]+\{a,b\}= \Big(a+[0,6] \Big)\cup \Big(b+[0,6]\Big).  
 \end{align}

\end{remark}

\begin{remark}
   It is well known that the spectrum of a family of \emph{ergodic} operators is almost surely non-random, dating back to L. Pastur \cite{pastur1980}. 
However, since the Sierpinski lattice $\G$ is not vertex transitive,  \eqref{eqn:AM} 
  is not realized as an ergodic family in the usual way by a natural vertex-transitive group of graph automorphisms (see e.g. \cite[Definition 3.4.]{aizenman2015random}). We thus do not know if there is a similar approach to show the deterministic nature of the spectrum.

 The more specific structure of the spectrum as described by the form of \eqref{eqn:AM-as-spectrum2} was determined for the Anderson model on $\Z^d$ in \cite{kunz1980,kirsch1982cmp} using the Weyl criterion. Theorem~\ref{thm:AM-as-spectrum} can be viewed as an extension of the result of \cite{kunz1980} from the $\Z^d$ lattice to the Sierpinski lattice $\G$ in certain cases. 
 
\end{remark}

 \begin{figure}[!ht]
	\centering
	\subfigure []{
		\includegraphics[width=7.0 cm]{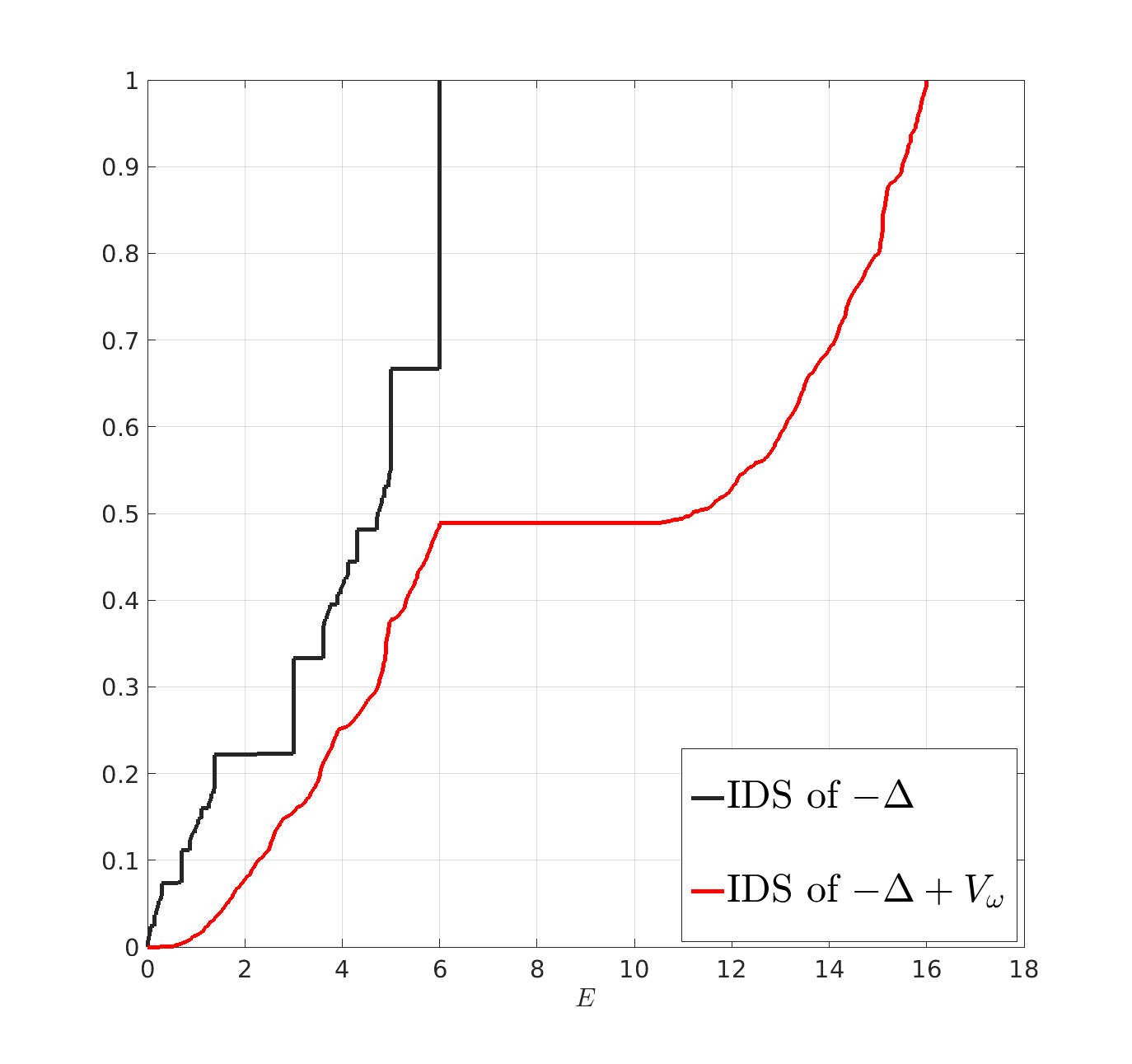}
	}
	\quad
        \subfigure[]{
		\includegraphics[width=7.0cm]{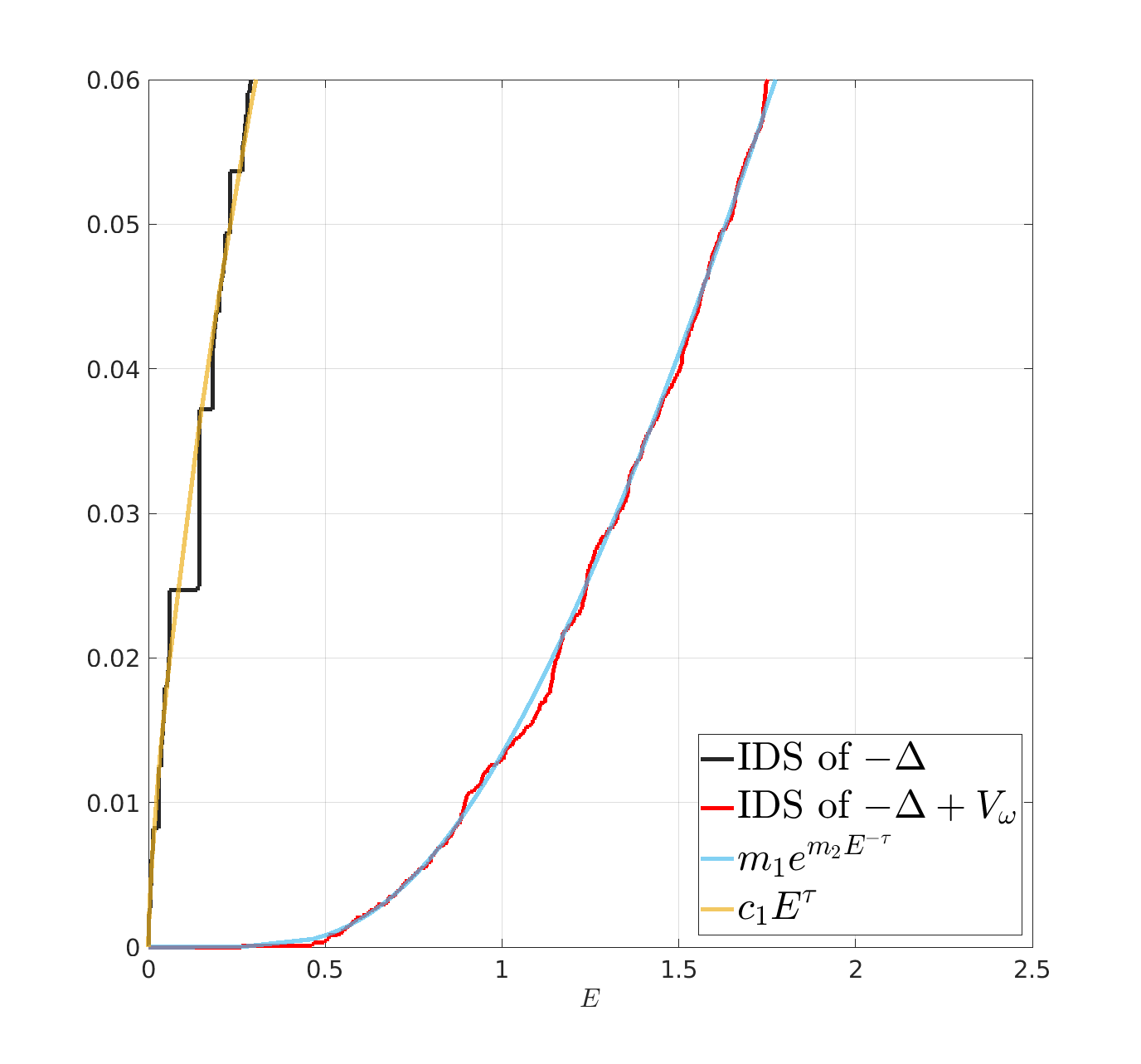}
	}
	\caption{
    (a) IDS of $-\Delta$ and $-\Delta+V_\omega$ with a 0-10 Bernoulli potential, for a finite gasket with sidelength $2^8$. The spectrum $\mathcal C=\sigma(-\Delta)$ is a Cantor subset of $[0,6]$. Clearly,  $\sigma(-\Delta+V_\omega)\subseteq[0,6]\cap[10,16]$. But it does not appear that $\sigma(-\Delta+V_\omega)\subseteq\mathcal C\cap (10+\mathcal C)$. 
    (b) The enlarged view of the bottom part of the IDS shows the Lifshitz tails in the Bernoulli case (red), along with a reference exponential function (light blue) where $\tau=\frac{\log 3}{\log 5}$, $m_1=1.38$, $m_2=-4.64$.
    Additionally, we plot the IDS for the standard Laplacian (in black) for comparison along with the reference line $c_1E^{\tau}$, with  $c_1=0.135$, in yellow.}\label{fig:gasket-ids}
\end{figure}

 
   Next, we study the integrated density of states of the random Schr\"odinger $H_\omega$, as the limit of a sequence of finite volume eigenvalue counting functions.  Given $L\in \N$, let  $B_{L}=B(O,2^L)\subseteq\G$    
    be the (graph metric) ball, centered at the origin $O=(0,0)$, with radius $2^L$. Let $H_\omega^{B_L}=\one_{B_L}H_\omega \one_{B_L}$ be the restriction $H_\omega$ on the finite-dimensional space $\ell^2(B_L)$. Denote the eigenvalue counting function below the energy $E$ of $H_\omega^{B_L}$ by 
    \begin{align}
        \mathcal N(E;H_\omega^{B_L})=\#\big\{\ {\rm eigenvalues}\ E'\ {\rm of}\ H_\omega^{B_L} \ {\rm such\ that\ }\ E'\le E\ \big\}.
    \end{align}

\begin{theorem}\label{thm:IDS-exist-lif}
Let $H_\omega=-\Delta+V_\omega$  be the Anderson model \eqref{eqn:AM} on the Sierpinski lattice $\G$.   Then there exists a non-random right continuous non-decreasing function  $N(E)$ such that almost surely, 
\begin{align}\label{eqn:IDS-exist-intro}
  N(E)   =   \lim_{L\to\infty}\frac{1}{|B_{L}|} \E \mathcal N(E; H_\omega^{ B_{L}})=\lim_{L\to\infty}\frac{1}{|B_{ L}|} \mathcal N(E; H_\omega^{ B_{ L}}) \ \   {\textrm{for all continuity points of}}\; N(E).  
    \end{align} 
    The limit $N(E)$  is defined to be the integrated density of states of $H_\omega$.

    If, in addition, suppose the common distribution $P_0$ of $\{V_\omega(x)\}_{x\in \G}$ satisfies
$\inf \supp P_0=0$ and 
$P_0([0,\eps])\ge C\eps^\kappa$ for some $C,\kappa>0$. 
 Then    
 \begin{align}\label{eqn:lif-SG-intro}
  \lim_{E\searrow 0}   \frac{\log \big|\log  N  (E)\big|}{\log E}=-\frac{\log 3}{\log5}. 
\end{align}
\end{theorem}

\begin{remark}
In general, to relate the finite volume approximations to the infinite operator, it is convenient to set $H_\omega^{B_L}$ to be zero in the
complement of $B_L$, corresponding to the so-called simple or zero Dirichlet boundary conditions. 
This is how one can interpret the operators in $H_\omega^{B_L}$ in \eqref{eqn:IDS-exist-intro}.
The finite volume operator can however take  many rather arbitrary choices of boundary conditions: free, Neumann, or wired in some way. We will state a more general version of the existence result \eqref{eqn:IDS-exist-intro} in Theorem~\ref{thm:IDS-exist-gen} in Section~\ref{sec:exi-ids}, where we see boundary conditions do not affect the limiting IDS. 
\end{remark}

\begin{remark}
We know by \eqref{eqn:AM-as-spectrum1} that the bottom of the spectrum $\sigma(-\Delta+V_\omega)$ is 0 assuming $\inf \supp V=0$. The limit \eqref{eqn:lif-SG-intro} is a weak version of what we expect to be a stronger form of the asymptotic behavior of the IDS near 0, 
\begin{align}\label{eqn:lif-gen}
    N(E)\sim C_1e^{-C_2E^{-\tau}},\ \tau=\frac{\log 3}{\log 5}, \ \ \ {\rm as}\ E\searrow 0. 
\end{align}
(See Figure~\ref{fig:gasket-ids}(b) for a numerical illustration.)
Such a drastic thinning tail of the IDS as in Eqs.~\eqref{eqn:lif-SG-intro} and \eqref{eqn:lif-gen} is referred to as Lifshitz tails. 
Note that for the free Laplacian $-\Delta$ on the Sierpinski gasket, the IDS $N(E)$ near the bottom $0$ of $\sigma(-\Delta)$ vanishes at the rate $N(E)\sim CE^{\tau}$  as $E\searrow 0$, as obtained in \cite{fuku1988,fuku1992}. In contrast to the polynomial behavior of the free Laplacian, the IDS for a random Schr\"odinger operator exhibits a different extreme behavior near the bottom such as \eqref{eqn:lif-gen}. The double-log asymptotic behavior \eqref{eqn:lif-SG-intro} will be proved in Section~\ref{sec:lif}.  We first discuss more background about Lifshitz tails to finish the introduction. 
\end{remark}


 \subsection{More background and historical work on the Lifshitz tails}\label{sec:background}
 
On the (continuum) Sierpinski gasket in $\R^2$, Lifshitz tails of the IDS were first proved for the Laplacian with Poisson obstacles in \cite{pietruska1991lifschitz}. Approximately the same time, similar Lifshitz tails were obtained for random Schr\"odinger operators on general nested fractals in $\R^d,d\ge2$ in \cite{shima1991lifschitz}.  
More recently, there are more generalizations in \cite{balsam2023density} for the differential case on 
nested fractals with good labeling properties.   All these works are for continuous/differential operators on the Sierpinski gasket or other continuous fractals. In this context, Eq.~\eqref{eqn:lif-SG-intro} of Theorem~\ref{thm:IDS-exist-lif} extends the Lifshitz tails to the Anderson model on the discrete/graphical Sierpinski lattice $\G$.  

The original Lifshitz tails phenomenon was first identified in the 1960s by I.M. Lifshitz  \cite{lif1965}. 
It has since been extensively studied with rigorous proof for various random models on $\R^d$ or $\Z^d$. 
We do not discuss all the related works here, but we refer readers to \cite[Section 6]{kirsch2007invitation} and  \cite[Chapter 4]{aizenman2015random} for a thorough review. Note that the original Lifshitz tail of the IDS $N(E)$ on $\R^d$ or $\Z^d$ is asymptotically $ N(E)\sim C_1e^{-C_2E^{-d/2}}$, as $ E\searrow0$ (assuming the bottom of the spectrum is at 0). The index $d/2$ is usually referred to as the Lifshitz exponent. One obtains a different Lifshitz exponent $\tau=\log3/\log5$ on the Sierpinski gasket (\cite{pietruska1991lifschitz,shima1991lifschitz}) and on the Sierpinski lattice \eqref{eqn:lif-SG-intro}. 
To relate $\tau$ to the exponent $d/2$, one denotes $\tau=d_s/2$, where $d_s$ is the `spectral dimension' of the gasket. 

Actually, there is a more intrinsic way to link the Lifshitz exponent to two other parameters in the so-called \emph{Heat Kernel Bound} $\mathrm{HK}(\alpha,\beta)$, a property for the free Laplacian on the corresponding space.  The Euclidean space/lattice $\R^d$ or $\Z^d$ satisfies $\mathrm{HK}(d,2)$, while the Sierpinski gasket or Sierpinski lattice satisfies $\mathrm{HK}(\log3/\log2,\log5/\log 2)$. In either case, we see that the Lifshitz exponent is given by the ratio of the two parameters $\alpha/\beta$. 
We are very interested in whether the Lifshitz singularity with exponent $\alpha/\beta$ holds for the Anderson model on more general graphs with the \emph{Heat Kernel Bound} $\mathrm{HK}(\alpha,\beta)$, without certain additional regularity/self-similarity/good labeling properties of the graph.

    For random differential/continuous operators on the Sierpinski gasket in $\R^2$ (or more generally, on nested fractals in $\R^d$), there are two main ways in the literature to prove the existence of the IDS and  Lifshitz tails. 
     \begin{itemize}
         \item  One way is the works of Pietruska-Paluba, Balsam, Kaleta, and Olszewski 
\cite{pietruska1991lifschitz,balsam2023density}.  The existence of the IDS is obtained by the convergence of the expected values of the underlying Laplace transforms. Then the Lifshitz tail is obtained by the long time behavior of the associated Laplace transform. 
         \item The other method is the work of  Shima \cite{shima1991lifschitz}.  The existence of the IDS was obtained directly as the limit of the finite volume IDS, by the law of large numbers. Then the Lifshitz tail of the IDS is obtained by the  Dirichlet--Neumann bracketing method, combined with the specific bound  by the Dirichlet form due to Kusuoka \cite{dobrushin1993lecture}.          
     \end{itemize}
    In this work, we extend the existence and the Lifshitz tail of the IDS of \cite{shima1991lifschitz} to the discrete  Sierpinski lattice $\G$,  using an adapted Dirichlet and Neumann bracketing method for $\G$.  There are numerous proofs of Lifshitz tails for different models using the method of large deviations.  
The idea of using Dirichlet--Neumann bracketing for Lifshitz tails
first appeared in a physics paper \cite{harris1973rigorous}.  The original Dirichlet and Neumann bracketing method was first established for
continuum models in $\R^d$ by Kirsch and Martinelli \cite{kirsch1983large}, with the advantage of being very close to Lifshitz's intuition. It was later extended by Simon \cite{simon1985lifschitz} to the $\Z^d$ setting, where there are some technical aspects special to the discrete model, see more discussion in \cite[\S]{kirsch2007invitation}, \cite[\S 4.3]{aizenman2015random}.
       The main technical difficulties of adapting the bracketing method to the Sierpinski lattice (which are not present for $\Z^d$ or continuum Sierpinski gasket cases) are: 
\begin{enumerate}[(i)]
\item There is not a natural disjoint partition of the Sierpinski lattice $\G$.   One has to carefully treat the overlap of the subdomains and the associated edge energy in order to bound the Hamiltonian. \item 
       The required bounds on the (low lying) eigenvalues of the Dirichlet or Neumann Laplacian on a finite Sierpinski triangle are not derived previously in the literature. These estimates might have independent interests in studying the spectrum of the associated Anderson model.
\end{enumerate}
       One of our goals is to describe some of these technical difficulties where Dirichlet--Neumann bracketing on the  Sierpinski lattice is not so
common,  honoring the bracketing method by  Kirsch-Martinelli--Simon and the adaption by Shima.  
 These results may also be useful to further study the Anderson localization and other open questions for random Schr\"odinger operators on the Sierpinski lattice, and on more general discrete graphs.

\subsection{Outline}
The rest of this article is organized as follows. 
\begin{itemize}
\item In Section~\ref{sec:pre}, we introduce background and preliminaries on the Sierpinski lattice. 
\item In Section~\ref{sec:spe}, we prove Theorem~\ref{thm:AM-as-spectrum}, the almost sure spectrum of the Anderson model on the Sierpinski lattice.
\item In Section~\ref{sec:exi-ids}, we study the partition of the Sierpinski lattice and show the  existence of the IDS under different boundary conditions. 
\item In Section~\ref{sec:lif}, we prove  the   Lifshitz tail \eqref{eqn:lif-SG-intro}, first the upper bound using the Neumann bracketing, and then the lower bound using the (modified) Dirichlet bracketing. 
\end{itemize}

Throughout the paper, constants such as $C$, $c$, and $c_i$ may change from line to line. 
We will use the notation $X\lesssim Y$ to mean $X\le cY$, and $X\gtrsim Y$ to mean $X\ge cY$, for some
 constant $c$ depending only on $\Gamma$. If $X\lesssim Y\lesssim X$, we may also write $X\approx Y$.

\section{Preliminaries}\label{sec:pre}

In this section, we collect several useful facts about Sierpinski lattices.   
We refer readers for more background about the Laplacian on Sierpinski lattices to \cite{shima1991eigenvalue,teplyaev1998spectral}, about random walks on general graphs to \cite{barlow2017random}, and about random Schr\"odinger operators to \cite{kirsch2007invitation,aizenman2015random}.

Let $\Gamma=(\G,\mathcal E)$ be the Sierpinski lattice defined in \eqref{eqn:SG-def}. $\Gamma$, or $\G$, is also called the full Sierpinski lattice with empty boundary/corner, while the one without the symmetric image $\G_n'$ is referred to as the  right half Sierpinski lattice with the boundary or corner $O=(0,0)$. 
Let $A\subseteq\G$ be a subset of vertices. The exterior  boundary of $A$ is   $\partial A=\{ x\not\in A: \exists y \in A \ \ {\rm with}\ \  x\sim y \}, $
 and the interior boundary is defined as $\partial^i A=\partial(\G\backslash A) $. The subset $A$ induces a subgraph  $\Gamma_A=(A,\mathcal E_A),$  where $\mathcal E_A=\{(x,y)\in \mathcal E: x,y\in A\}$.  
When there is no ambiguity, we may identify a graph or a subgraph with its vertex set and vice versa. For example,  we may call either $\Gamma_A$ or  just $A$ a subgraph of the Sierpinski lattice. We may also abuse the notation and write $x\in \Gamma_A$ if $x\in A$. 
The vertex degree of $x$, $\deg(x)=\#\{y: x  \sim y\}$ is the number of neighbors of $x$. Notice on the full Sierpinski lattice, $\deg(x)\equiv 4$, while on the right half Sierpinski lattice, $\deg(O)=2$ and $\deg(x)=4$ for $x\neq O$. 
As usual, a ball in $\Gamma$,  centered at $x$, with radius $r$, is defined as $B(x,r)=\{y:d(x,y)\le r\}$, where the natural metric $d(x,y)=$ length of the shortest path from $x$ to $y$.  For any subset $A\subseteq\G$, we denote by $|A|=\#\{x: x\in A\}$ the cardinality of $A$. 

For any $L\in \N$, $\G_L$ in \eqref{eqn:SG-def} induces a subgraph. The entire Sierpinski lattice consists of infinitely many translations of $\G_L$, denoted as $\{\G_{L,j}\}_{j=1}^\infty$, glued together at corner points. We  call either $ \G_L$ or any of its translations a $2^L$-triangle. The extreme points/vertices of $\G_{L,j}$ are the three vertices of the biggest triangle, i.e., the interior boundary of $\G_{L,j}$. We say that two $2^L$-triangles are adjacent if they share an extreme point. Due to the recursive relation \eqref{eqn:T-rec}, 
\begin{align}
    |\G_L|=\frac{1}{2}(3^{L+1}+3),\ \ \ L=0,1,2,\cdots. 
\end{align}
Suppose $L\ge \ell$. Then the $2^L$-triangle $\G_L$ consists of $3^{L-\ell}$ many $2^\ell$-triangles $ \G_{\ell,j},j=1,\cdots,3^{L-\ell}$. All these $ \G_{\ell,j}$ are subgraphs isometric to $\G_\ell$.

We consider functions on the vertices $\G$, which will be denoted by the function space 
$C(\G):=\C^{\G}=\{f: \G\to  \C \}$.  
 The space $\ell^2(\G)$ is defined via the $\ell^2$ norm induced by the usual (non-weighted) inner product  $\langle f,g\rangle:=\sum_{x\in \G} f(x) \overline{ g(x)}. $
The subspace  $\ell^2(A)$ is defined accordingly for any subset $A\subseteq\G$. The  Laplacian $\Delta$ on the full Sierpinski lattice $\G$ is defined as in \eqref{eqn:Lap}. It is a bounded nonpositive selfadjoint operator in  $\ell^2(\G)$. For $f,g\in \ell^2(\G)$, the following Discrete Gauss–Green theorem holds:
\begin{align}\label{eqn:gauss-green}
    \sum_{x\in \G}g(x)\Delta f(x)=-\frac{1}{2}\sum_{x\in \G}\sum_{\substack{y\in \G\\
    y\sim x}}\big(f(x)-f(y)\big)\big(g(x)-g(y)\big).
\end{align}
The probabilistic Laplacian $\Delta_p$ on the full Sierpinski lattice $\G$ is defined to be $\Delta_p=\frac{1}{4}\Delta$ since $\deg(x)\equiv 4,x\in \G$. The  structure of the  spectrum   $\sigma(\Delta_p)=\frac{1}{4}\sigma(\Delta)$ was fully determined in \cite{teplyaev1998spectral}.

        \section{Spectrum of the Anderson model on the Sierpinski lattice}\label{sec:spe}
    
In this section we prove Theorem~\ref{thm:AM-as-spectrum}.
\begin{proof}[Proof of Theorem~\ref{thm:AM-as-spectrum}]
 Suppose $ \supp P_0  \subseteq[a,b]$ for some $a\le b$, so that almost surely $a\le V_\omega \le b$. By the standard power series expansion argument of self-adjoint operators (see e.g. \cite[Problem~3.6]{teschl2014mathematical}), for  $a, b$ finite  we obtain the right hand side inclusion of \eqref{eqn:AM-as-spectrum1}, 
    \[  \sigma(-\Delta+V_\omega)\subseteq\sigma(-\Delta)+[a,b]. \] 
 If $a$ or $b$ is infinite, the inclusion is immediate since the right hand side is then an interval $(-\infty,\infty)$, $[a,\infty)$, or $(-\infty,b+6]$.
 The same argument as in the $a, b$ finite case implies \eqref{eqn:AM-as-spectrum3}, using the fact that $\sigma(-\Delta)\subseteq[0,6]$ and switching the role of $-\Delta$ and $V_\omega$.      
     
    It remains to prove the left hand side of \eqref{eqn:AM-as-spectrum1}. The outline follows from the proof of the $\Z^d$ case, see e.g. \cite[Theorem 3.9]{kirsch2007invitation}. There are two key ingredients needed for the Sierpinski lattice, stated in the following two claims.

\begin{claim}\label{clm:Vconst}
There is a set $\Omega_0$ of probability one such that the following is
true:     For any $\omega\in \Omega_0$, any  $\mu\in  \supp P_0 $, $\ell>0$ and $\varepsilon>0$,  there exists a $2^\ell$-triangle $\G_\ell\subseteq\G$ such that 
\begin{align}
    \sup_{x\in \G_{\ell}}|V_\omega(x)-\mu|<\varepsilon. 
\end{align}
\end{claim}
\begin{remark}
    This is a ``Sierpinski lattice analogue'' of the result on $\Z^d$. It is important that the full probability set is independent of $\mu,\ell$ and $\eps$. We only need the existence of some (one) triangle for any size $\ell$. The conclusion can be made stronger by requesting infinitely many triangles;  see the analog for $\Z^d$ in e.g. \cite[Proposition 3.8]{kirsch2007invitation}.
\end{remark}

\begin{claim}\label{clm:ef-move}

Eigenvalues of $-\Delta$ are dense in the spectrum $\sigma(-\Delta)$.
Each eigenvalue has infinitely many compactly supported eigenfunctions. For any eigenfunction $\varphi$ supported on some $2^\ell$-triangle $\G_\ell$ and  for any other $2^\ell$-triangle $\G_{\ell,1}$,
there is a translation of $\varphi$, denoted as $\psi$, supported on $  \G_{\ell,1}$, belonging to the same eigenvalue. 
\end{claim}

Claim~\ref{clm:Vconst} is a consequence of the Borel--Cantelli lemma, and guarantees that the potential can be arbitrarily close to any  constant on a ``far away'' triangle.  Claim~\ref{clm:ef-move} is essentially a rephrase of \cite[Theorem 2]{teplyaev1998spectral}, which allows one to move a compactly supported eigenfunction anywhere on the lattice. The detailed proofs of the two claims are left to the end of the section. We first use them to complete the proof of Theorem~\ref{thm:AM-as-spectrum}.

 Suppose $\lambda\in  \sigma(-\Delta)$ and $\mu\in  \supp P_0 $. We will construct a Weyl sequence of $-\Delta+V_\omega$ associated with $\lambda+\mu$.  By Claim~\ref{clm:ef-move}, there is a sequence of 
compactly supported eigenfunctions $\varphi_k$ associated with eigenvalues\footnote{
If $\lambda$ itself is an eigenvalue with an eigenfunction $\varphi$, then we take $\lambda_k\equiv \lambda$ and $\varphi_k\equiv \varphi$.} $\lambda_k$ such that $|\lambda_k-\lambda|<1/k$ and $\|\varphi_k\|=1$ for all $k$. 
For each $\varphi_k$, since it is compactly supported, we assume $\supp \varphi_k$ is contained in some $2^{\ell_k}$-triangle $\G_{\ell_k}$.  
Now take $\omega\in \Omega_0$ given as in Claim~\ref{clm:Vconst}. Then for $\mu\in  \supp P_0 $,  $\ell=\ell_k$ and $\eps=1/k$,  there is a $2^{\ell_k}$-triangle $  \G_{\ell_k,1}$ (not necessarily the same as $\G_{\ell_k}$ where $\varphi_k$ is supported) such that 
\begin{align}
    \sup_{x\in   \G_{\ell_k,1}}\big|V_\omega(x)-\mu\big|<\frac{1}{k}. 
\end{align}
  
Next, we use Claim~\ref{clm:ef-move} to move $\varphi_k$ from $\G_{\ell_k}$ to $ \G_{\ell_k,1}$ which gives a new eigenfunction $\psi_k$ of $-\Delta$ satisfying $\|\psi_k\|=1$, $-\Delta\psi_k=\lambda_k \psi_k$, and $\supp \psi_k\subseteq \G_{\ell_k,1}$.

Then almost surely one has
 \begin{align}
     \|(-\Delta +V_\omega-(\lambda+\mu)\psi_k\|\le &   \|(-\Delta -\lambda)\psi_k\|+\|(V_\omega-\mu)\psi_k\|  \\ 
     \le & |\lambda_k-\lambda|\cdot \|\psi_k\|+ \sup_{x\in   \G_{\ell_k,1}} \big|V_\omega(x)-\mu\big| \cdot\|\psi_k\|\le \frac{2}{k}. 
 \end{align}
 Hence $\psi_k$ is a Weyl sequence of $-\Delta +V_\omega$ associated with $\lambda+\mu\in \sigma(-\Delta )+ \supp P_0 $. The   Weyl criterion (see e.g. \cite[Theorem VII.12]{reed1981functional} or \cite[Lemma 6.17]{teschl2014mathematical}) implies $\lambda+\mu \in \sigma(-\Delta+V_\omega)$ which completes the proof of  \eqref{eqn:AM-as-spectrum1}.
\end{proof}

The rest of this section contains the proofs of the two claims.
We first show 
\begin{proof}[Proof of Claim~\ref{clm:Vconst}]
Fix $\ell>0$,  $\mu\in  \supp P_0 $, and $\eps>0$. Let $\{\G_{\ell,j}\subseteq\G\}_{j\ge 1}$ be infinitely many disjoint $2^\ell$ triangles, and let $\mathcal E_j=\{\omega:|V_\omega(x)-\mu|<\varepsilon,\ x\in \G_{\ell,j}\}$. Since $\{V_\omega(x)\}$ are i.i.d.,  the probability of $\mathcal E_j$ is $\P(\mathcal E_j)= \P\big(\mu-\varepsilon,\mu+\varepsilon \big)^{|\G_{\ell,j}|}=p_{\mu,\ell,\varepsilon}>0,$
    which is independent of $j$. Hence, $\sum_{j=1}^\infty\P(\mathcal E_j)=\infty.$ 
  Since all $\G_{\ell,j}$ are disjoint and all the events $\mathcal E_j$ are independent, by the (second) Borel--Cantelli lemma (see e.g. \cite[Theorem 3.6]{kirsch2007invitation}), the set 
  \[\Big\{\omega\ | \ \omega\ {\textrm{belongs to infinitely many}}\ \mathcal E_j \Big\} \]
  has probability one, which implies that \begin{align}
      \Omega_{\ell,\varepsilon,\mu}=\big\{\omega\ |\ {\rm for\ some\ }2^\ell\ {\rm triangle}\ \G_{\ell}\subseteq\G: \sup_{x\in \G_{\ell}}|V_\omega(x)-\mu|<\varepsilon \ \big\}
  \end{align}
   has probability one.    Since the set $ \supp P_0\subseteq\R $ contains a countable dense set $\mathcal C_0$, then the countable intersection set 
   \begin{align}
      \Omega_0=\bigcap_{\substack{\ell\in \N, n\in \N\\
      \mu\in \mathcal C_0}}\Omega_{\ell,1/n,\mu} 
   \end{align}
   also has probability one, and satisfies the requirement of the claim.   
\end{proof}

Now we  verify  Claim~\ref{clm:ef-move}. 
\begin{proof}[Proof of Claim~\ref{clm:ef-move}]
By \cite[Theorem 2]{teplyaev1998spectral}, the spectrum of $\Delta$ is  $\sigma(\Delta)=-4(\mathcal D \cup \mathcal J)$, where $\mathcal D=\{-3/2\}\cup \big(\bigcup_{m=0}^\infty R_{-m}\{-3/4\}\big)$,  for $R(z)=z(4z+5)$ and $R_{-m}A$ the preimage of a set $A$ under the $m$-th composition power of $R$, is part of the eigenvalues of $-\Delta$. The set $\mathcal J$ is the Julia set of $R$ which  coincides with the set of limit points of $\mathcal D$. 
Hence, $ 4\mathcal D$ is dense in $\sigma(-\Delta)$.
Also by \cite[Theorem 2]{teplyaev1998spectral}, 
any eigenvalue  of $-\Delta$ has infinitely many compactly supported eigenfunctions.

It remains to verify the latter half, which allows us to move a compactly supported eigenfunction anywhere on the Sierpinski lattice. 
This follows from the repeated structure of the Sierpinski lattice.
More precisely, if $\varphi$ is an eigenfunction of $-\Delta$ supported on a $2^\ell$-triangle $\G_{\ell}$, then one can translate $\varphi$ to any other $2^\ell$-triangle $\G_{\ell,1}$ to create another eigenfunction with the same eigenvalue. This uses that $\deg(x)=4$ for all $x\in\G$, so that the restrictions of $-\Delta$ on $\G_\ell$ and $\G_{\ell,1}$ are exactly the same (as a finite dimensional matrix).
\end{proof}

\begin{remark}\label{rem:full-half}
The proof of Claim~\ref{clm:ef-move} is for the combinatorial Laplacian $\Delta$, or equivalently (4 times) the probabilistic Laplacian $4\Delta_p$,   on  the full Sierpinski lattice $\G=\bigcup_n(\G_n\cup\G_n')$.     
A similar argument (with accounting for the origin) applies to $-\Delta_p$ on the right half Sierpinski lattice, using the spectrum structure of $\sigma(\Delta_p)$ obtained in \cite{teplyaev1998spectral}. 
 \end{remark}

\section{Existence of the IDS for random Schr\"odinger operators on the Sierpinski lattice.}\label{sec:exi-ids}

Let   $H_\omega=-\Delta+V_\omega$ be the Anderson model as in \eqref{eqn:AM}.
 For simplicity, we omit the subscript $\omega$ and write $H=H_\omega$ and $V=V_\omega$. We first discuss restrictions of $H$ to a finite-dimensional subspace of $\ell^2(\G)$ with different boundary conditions.  
For any finite subset $A\subseteq\G$, we need to consider the following three Laplacians with different boundary conditions:
\begin{itemize}
    \item Simple boundary condition (the usual zero Dirichlet boundary condition): for $f\in \ell^2(A)$, 
    \begin{align}\label{eqn:Lap-simple}
        -\Delta^{A}f(x)=-\Delta^{A,S}f(x)= \deg_{\G\backslash A}(x) f(x)+\sum_{\substack{y\in A \\ y\sim x}}\big(f(x)-f(y)\big)=\deg(x)f(x)-\sum_{\substack{y\in A \\ y\sim x}}f(y).
    \end{align}
    \item Neumann boundary condition: for $f\in \ell^2(A)$, 
    \begin{align}\label{eqn:Lap-N}
        -\Delta^{A,N}f(x)= \sum_{\substack{y\in A \\ y\sim x}}\big(f(x)-f(y)\big)=\deg_A(x)f(x)-\sum_{\substack{y\in A \\ y\sim x}}f(y).
    \end{align}
    \item Modified Dirichlet boundary condition: for $f\in \ell^2(A)$, 
    \begin{align}\label{eqn:Lap-D}
        -\Delta^{A,D}f(x)=2\deg_{\G\backslash A}(x) f(x)+\sum_{\substack{y\in A \\ y\sim x}}\big(f(x)-f(y)\big)=\big(2\deg(x)-\deg_A(x)\big)f(x)-\sum_{\substack{y\in A \\ y\sim x}}f(y).
    \end{align}
\end{itemize}
The corresponding Schr\"odinger operator is $H^{A,\bullet}=-\Delta^{A,\bullet}+V^A$, where $\bullet$ is one of the above three boundary conditions and $V^A$ is just the restriction of $V$ to $A$. 

The zero/simple boundary Laplacian    $\Delta^{A}=\Delta^{A,S}$ corresponds to application of $\Delta$ on the subspace   $\{f\in \ell^2(\G):f(x)=0,x\notin A\}$, and is associated with the simple random walk killed upon exiting $A$.  The modified Dirichlet boundary Laplacian $\Delta^{A,D}=\Delta^{A}-\deg_{\G\backslash A}$ is a Dirichlet-type operator with  zero/simple boundary conditions 
on $\G\backslash A$, modified however at the interior boundary vertices of $A$, where it penalizes such boundary vertices. The Neumann boundary Laplacian $\Delta^{A,N}$ is the same as the regular graph Laplacian on the subgraph induced by $A$ (i.e. without consideration for the larger graph $\G$), and thus satisfies the same type of Gauss--Green theorem as the infinite lattice \eqref{eqn:gauss-green}: 
\begin{align}\label{eqn:gauss-green-finite}
     \ipc{f}{ -\Delta^{A,N}   f}_{\ell^2(A)} =&\, \frac{1}{2}\sum_{\substack{x,y \in A \\ x\sim y } }\big(f(x)-f(y)\big)^2. 
\end{align}

\begin{remark}
  Notice that in the interior, if $x\notin \partial^iA$,
 then $\deg(x)=\deg_A(x)$ and the three Laplacians behave in the same way. The simple and modified Dirichlet boundary conditions only add extra diagonal terms (corresponding to vertex degrees) to the interior boundary vertices. This will allow us to bound the extra edge energy terms in the quadratic form when we partition the graph into subgraphs. These modified vertex degree operators were first used in \cite{simon1985lifschitz} for $\Z^d$ case, demonstrating some technical aspects in the bracketing method
special to the discrete case; see also the discussion in \cite[\S 4.3]{aizenman2015random}. We will see how these boundary conditions will affect the energy partition on the Sierpinski lattice in Section \ref{sec:lif}.   
\end{remark}

For any $L\in \N$, let $\Gamma_{L,j}=(\G_{L,j},\mathcal E)$ be a $2^L$-triangle. The collection $\{\G_{L,j}\}_j$ forms a non-disjoint cover  of $\G$,
 \begin{align}\label{eqn:Part1-SG}
     \mathcal P=\{\G_{L,j}\}_{j\ge 1} ,\ \  \G =\bigcup_{j\ge 1} \G_{L,j}. 
 \end{align}
Two adjacent triangles $\G_{L,j}$  and $\G_{L,j'}$ share  only one extreme vertex. The overlap of $\mathcal P$ will not affect the Neumann bracketing side or the upper bound of the eigenvalue counting. For the Dirichlet bracketing side, we will need the following surgeries on $\mathcal P$ to deal with the overlap of two adjacent triangles.  
For any $\G_{L,j}\in \mathcal P$, let $\partial ^i\G_{L,j}=\{o_1,o_2,o_3\}$
be the three extreme vertices. Denote by $\wt \G_{L}=\G_{L}\backslash (\partial ^i\G_{L})$ the truncated $2^\ell$-triangle associated with $\G_\ell$ (Figure~\ref{fig:wtV}). Consider the following disjoint collection
\begin{align}\label{eqn:Part2-SG}
     \wt{\mathcal  P} =\{\wt \G_{L,j}\}_{j\ge 1},  \ \  \G=(\bigcup_{j\ge 1} \wt \G_{L,j})\cup \mathcal R,
 \end{align}
where $\mathcal R$ is the collection of all the extreme vertices of all $\G_{L,j}\in \mathcal P$. The associated subgraph is denoted as $\wt \Gamma_{L,j}=(\wt \G_{L,j},\mathcal E)$, see Figure \ref{fig:wtV}. Then $\{ \wt{\mathcal P},\mathcal R\}$ forms a partition of $\G$.   

\begin{figure}
    \centering
    \includegraphics[width=0.45\textwidth]{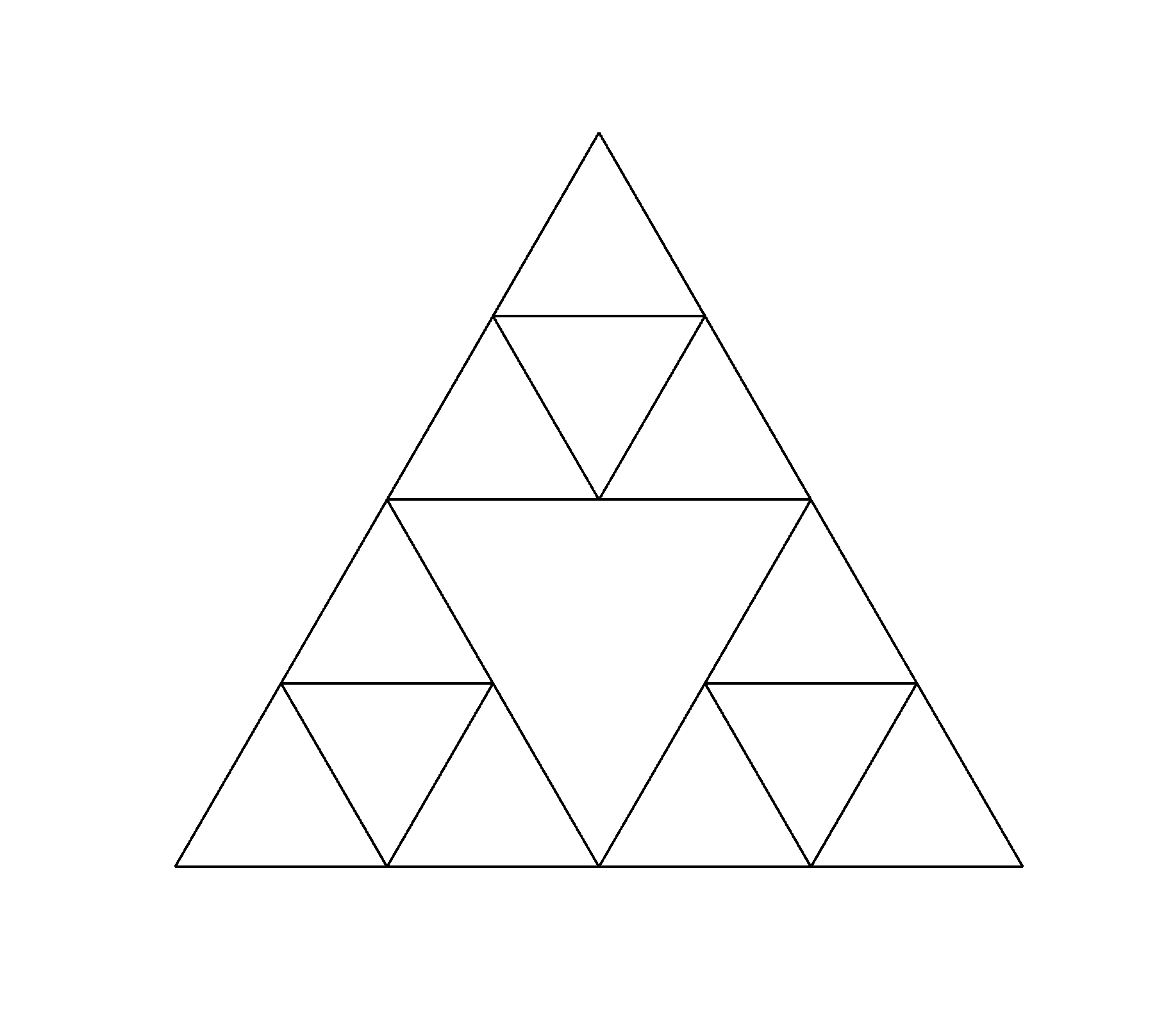}
    \includegraphics[width=0.45\textwidth]{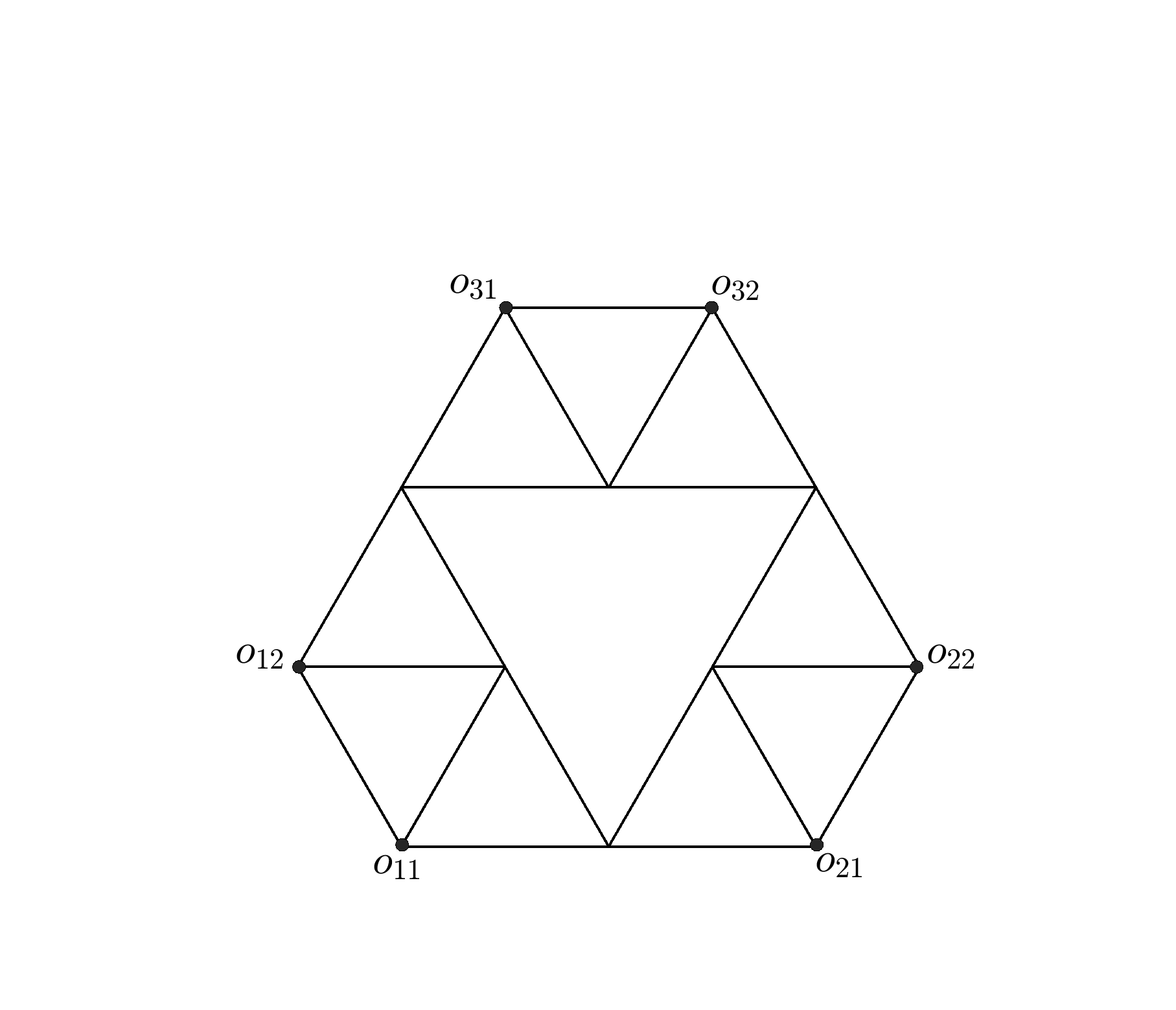}
    \caption{The left figure is the subgraph induced by $\G_2$. Removing the three extreme vertices in $\G_2$ leads to the truncated triangle $\wt \G_2$ on the right. The dashed edges are also removed from the subgraph induced by $\wt \G_2$.}
    \label{fig:wtV}
\end{figure}

We now denote by \[\mathcal N(E;X)=\#\{{\rm eigenvalues}\ E'\ {\rm of}\ X \ {\rm such\ that\ }\ E'\le E\}\] the eigenvalue counting function of an operator $X$ below the energy $E$. 
\begin{lemma}\label{lem:XL-induc} Let $X^L,Y^L$ be  any two of the operators in the set $\bigcup_{\bullet\in\{S,N,D\}}\{H^{\G_L,\bullet},H^{\wt\G_L,\bullet}\}$,
where $\bullet\in\{S,N,D\}$ is one of the three boundary conditions \eqref{eqn:Lap-simple}--\eqref{eqn:Lap-D}. Then 
    \begin{align}\label{eqn:XY}
        \big| \mathcal N(E; X^L)-\mathcal N(E; Y^L)\big|\le 9,
    \end{align}
and  
   \begin{align}\label{eqn:XL}
        \big| \mathcal N(E; X^L)-\sum_{i=1}^3\mathcal N(E; X^{L-1,i})\big|\le 30,
    \end{align}
    where $\{X^{L-1,i}\}_{i=1}^3$ are the corresponding Schr\"odinger operators on the three $2^{L-1}$-triangles contained in $\G_L$, with the same boundary condition as $ X^L $.  

\end{lemma}
The bounds in \eqref{eqn:XY} and \eqref{eqn:XL} essentially follow from the min-max principle, using applications to rank-one perturbations or orthogonal projections between matrices. We give the proof here, referencing some standard technical lemmas included in Appendix~\ref{sec:min-max}. We will see in the proof that the general upper bounds in the right hand side of \eqref{eqn:XY} and \eqref{eqn:XL} hold for any choice of the boundary conditions, and that there are better estimates for a specific $X^L$. 
\begin{proof}
The three boundary conditions  \eqref{eqn:Lap-simple}--\eqref{eqn:Lap-D} differ only at the 3 extreme vertices of the triangle $\G_L$. In other words, the operators $H^{\G_L,D},H^{\G_L}, H^{\G_L,N}$ all act on $\ell^2(\G_L)$ in the same way,
except for the 3 diagonal terms  at the 3 extreme vertices. Then by a perturbation argument (Lemma~\ref{lem:A1} Eq.~\eqref{eqn:NH-diag-pert}), for $X^L,Y^L$ being  any two of $\{H^{\G_L,D},H^{\G_L}, H^{\G_L,N}\}$, we have 
  $ \big| \mathcal N(E; X^L)-\mathcal N(E; Y^L)\big|\le 3$.

Next, we link the counting between $\G_L$ and its truncation $\wt \G_L$. 
 By the definition of the simple boundary condition \eqref{eqn:Lap-simple}, $H^{\wt\G_L}$ is the orthogonal projection (restriction) of $H^{\G_L,N}$ onto the subspace $\ell^2(\wt \G_L)$,
 since the diagonal vertex degree terms are 4 in all cases for non-extreme vertices.
 By Lemma~\ref{lem:A1} Eq.~\eqref{eqn:NH-cauchy-inter}, 
 \begin{align}\label{eqn:NH+2}
    \mathcal N(E;H^{\wt\G_L})\le  \mathcal N(E;H^{\G_L,N})\le  \mathcal N(E;H^{\wt\G_L})+3.
\end{align}
On the other hand, when we remove the extreme vertices $o_i$, $i=1,2,3$, from $\G_L$, we create 6 new interior boundary points in $\wt \G_L$, denoted as $o_{ij}\in \wt \G_L$, $i=1,2,3$, $j=1,2$ with $o_{ij}\sim o_i$ (see e.g. Figure~\ref{fig:wtV}). When we count eigenvalues on $\wt \G_L$ corresponding to the different boundary conditions, we only need to consider the differences in the vertex degree at these 6 points $o_{ij}$. Similar as on the $\G_L$,   the operators $H^{\wt\G_L,D},H^{\wt\G_L}, H^{\wt\G_L,N}$  act on $\ell^2(\wt \G_L)$  with exactly the same matrix elements, except for the 6 diagonal terms at  $o_{ij}\in \wt \G_L$. Then again by \eqref{eqn:NH-diag-pert}, for $X^L,Y^L$ being  any two of $\{H^{\wt\G_L,D},H^{\wt\G_L}, H^{\wt\G_L,N}\}$,  we obtain $ \big| \mathcal N(E; X^L)-\mathcal N(E; Y^L)\big|\le 6$. Together with \eqref{eqn:NH+2}, this yields
  \eqref{eqn:XY}. 

\begin{figure}
    \centering
     \includegraphics[width=0.5\linewidth]{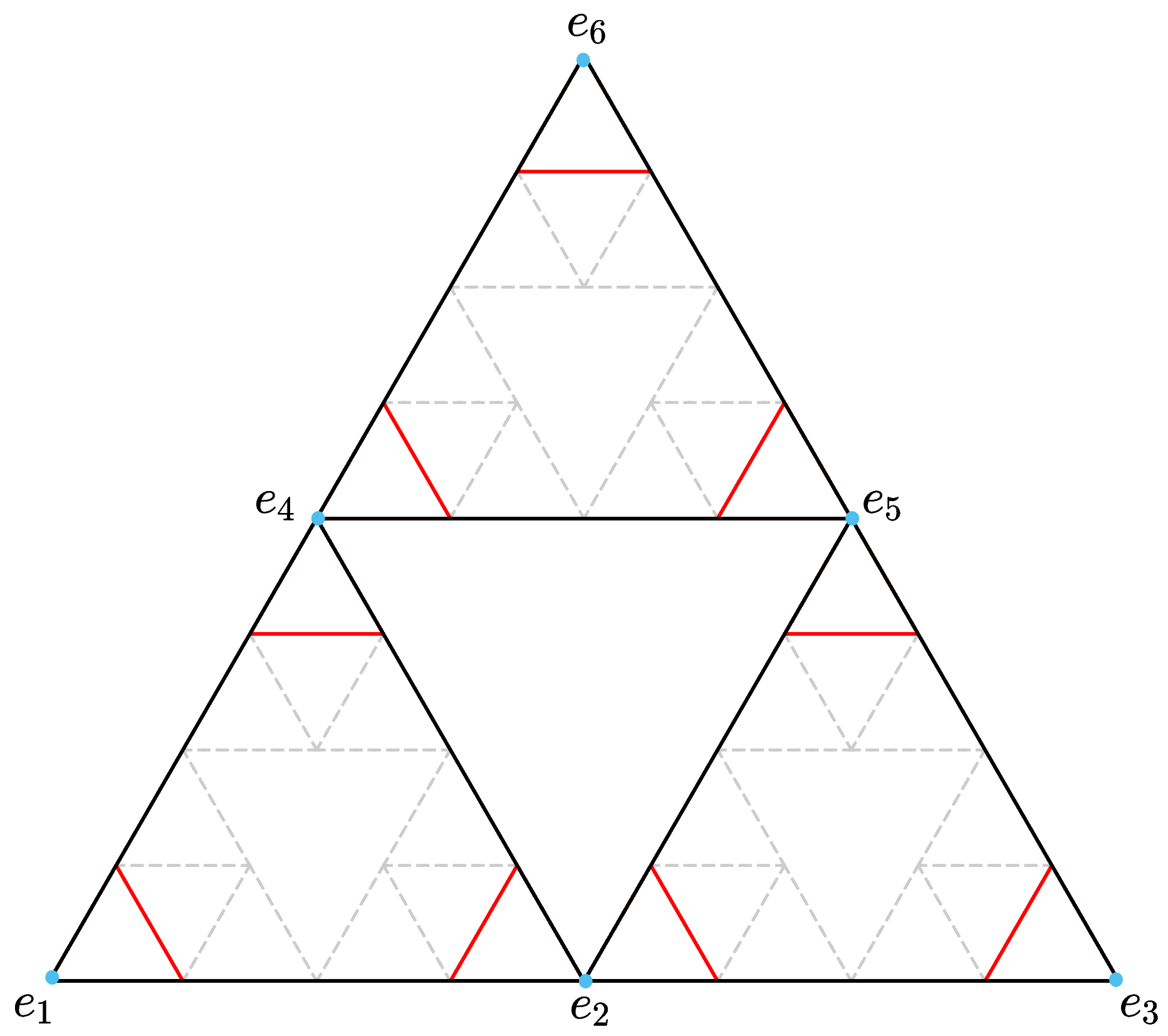}
    \caption{The $2^3$-triangle $\G_3$ (the big triangle) consists of three $2^2$-triangle $\{\G_{2,i}\}_{i=1}^3$. The three $2^2$-triangles have totally 6 extreme vertices $\{e_j\}_{j=1}^6$ (in blue).  
    Remove these extreme vertices, then the truncated triangles $\wt \G_{2,i}$ are disjoint, and $(\cup_{i=1}^3\wt \G_{2,i})\cup\{e_j\}_{j=1}^6$ form a disjoint partition of $\G_3$.}
    \label{fig:Vpart}
\end{figure}

For the upper bound in \eqref{eqn:XL}, we write the decomposition $\G_L=\cup_{i=1}^3\G_{L-1,i}$ where each $\G_{L-1,i}$ is a $2^{L-1}$-triangle isometric to $\G_{L-1}$. These three triangles only (pairwise) intersect at their three extreme vertices, $\cup_{i\neq j}\big(\G_{L-1,i}\cap\G_{L-1,j}\big)=\{e_1,e_2,e_3\}$ (see Figure~\ref{fig:Vpart}).   We consider only the Neumann boundary condition on each of these three triangles, since we can then use \eqref{eqn:XY} for the other boundary conditions. By the definition of the Neumann Laplacian \eqref{eqn:Lap-N},  we have
\begin{align}\label{eqn:Lap-N-quad}
    \ipc{f}{-\Delta^{\G_L,N} f}_{\ell^2(\G_L)}=\sum_{i=1}^3\ipc{f}{-\Delta^{\G_{L-1,i},N} f}_{\ell^2(\G_{L-1,i})}. 
\end{align}
Set $\wt V(x):=V(x)$ if $x\neq e_i$ and $\wt V_{e_i}:=2V_{e_i}$. Since the three $2^{L-1}$-triangles only overlap once at each $e_i$, then  
\begin{align}\label{eqn:Lap-N-quad-sch}
   \ipc{f}{\big(-\Delta^{\G_L,N}+\wt V^{\G_L}\big) f}_{\ell^2(\G_L)}=\sum_{i=1}^3\ipc{f}{\big(-\Delta^{\G_{L-1,i},N}+V^{\G_{L-1,i}} \big)f}_{\ell^2(\G_{L-1,i})}.   
\end{align}
 Hence, by Lemma~\ref{lem:NH<NH12} Eq.~\eqref{eqn:NH<H12},  we obtain
\begin{align}\label{eqn:412}
    \mathcal N(E; -\Delta^{\G_L,N}+\wt V^{\G_L})\le \sum_{i=1}^3\mathcal N(E; -\Delta^{\G_{L-1,i},N}+V^{\G_{L-1,i}}) .  
\end{align}
Since $\wt V^{\G_L}$ is an diagonal perturbation of $ V^{\G_L}$ with rank 3, then by Lemma~\ref{lem:A1} Eq.~\eqref{eqn:NH-diag-pert}, 
\begin{align}\label{eqn:413}
     \mathcal N(E; -\Delta^{\G_L,N}+\wt V^{\G_L})\ge   \mathcal N(E; -\Delta^{\G_L,N}+  V^{\G_L})-3=\mathcal N(E; H^{\G_L,N})-3. 
\end{align}
Putting \eqref{eqn:412} and \eqref{eqn:413} together with \eqref{eqn:XY} gives 
\begin{align}
    \mathcal N(E; X^L)\le \sum_{i=1}^3\mathcal N(E; X^{ L-1,i})+30 ,
\end{align}
for $X^L$ being any of $\{H^{\G_L,N},H^{\wt\G_L,D},H^{\wt\G_L}, H^{\wt\G_L,N}\}$. 

For the lower bound in \eqref{eqn:XL}, we need to split $\G_L$ into disjoint (truncated) smaller triangles in order to apply Lemma~\ref{lem:A3} (see Figure~\ref{fig:Vpart}).  Let $\G_L=\cup_{i=1}^3\G_{L-1,i}$ be as above. Let $\wt \G_{L-1,i}\subseteq\G_{L-1,i},i=1,2,3$ be the three truncated $2^{L-1}$ triangles isometric to $\wt \G_{L-1}$ as in \eqref{eqn:Part2-SG}. Denote by $\mathcal R=\{e_j\}_{j=1}^6:=\G_L\backslash\cup_{i=1}^3\wt \G_{L-1,i}$, where $e_j$ are   the 6 extreme vertices of $ \{\G_{L-1,i}\}_{i=1}^3$. 
Let $\mathcal H_i=\ell^2(\wt \G_{L-1,i}),i=1,2,3$ and $\mathcal H_0=\ell^2(\mathcal R)$.  Hence, $\mathcal H_i,i=0,1,2,3$ are orthogonal to each other and $\oplus_{i=0}^3  \mathcal H_i =\ell^2(\G_L)$. 

When we split the graph $ (\G_L,\mathcal E)$ into subgraphs $(\wt \G_{L-1,i},\mathcal E)$ (and the 6 isolated vertices $e_j$), we removed 18 edges connecting $\wt \G_{L-1,i}$ and $e_j$. 
Starting similarly to \eqref{eqn:Lap-N-quad}, we have, using \eqref{eqn:gauss-green},
\begin{align*}
    \ipc{f}{ -\Delta^{\G_L,N}   f}_{\ell^2(\G_L)}  
=& \sum_{i=1}^3\ipc{f}{ -\Delta^{\wt\G_{L-1,i},N}   f}_{ \mathcal H_i } +   \frac{1}{2}\sum_{\substack{x\sim y \\ x\, {\rm or}\,  y \in \mathcal R} } \big(f(x)-f(y)\big)^2 \\
 \le &  \sum_{i=1}^3\ipc{f}{ -\Delta^{\wt\G_{L-1,i},N}   f}_{ \mathcal H_i }  +  \sum_{\substack{x\sim y \\ x\, {\rm or}\,  y \in \mathcal R} } \big(f(x)^2+f(y)^2 \big)   \\
 \le  &  \sum_{i=1}^3\ipc{f}{ -\Delta^{\wt\G_{L-1,i},N}   f}_{ \mathcal H_i }  +\sum_{i=1}^3 
  \sum_{\substack{x\in \wt\G_{L-1,i} \\
 x\sim  \mathcal R}} 
 2f(x)^2 +8\sum_{j=1}^6 f(e_j)^2  \\
 \numberthis=& \sum_{i=1}^3\ipc{f}{ -\Delta^{\wt\G_{L-1,i},D}   f}_{ \mathcal H_i }   +   8\sum_{j=1}^6 f(e_j)^2. \label{eqn:energy-bound1}
\end{align*}
In the last line, we used that if $ x\in \wt\G_{L-1,i}$ and $x\sim\mathcal R$ (i.e. $ x\sim y$ for some $y\in  \mathcal R  $), then $x$ is an interior boundary point of $\wt\G_{L-1,i}$. Since 
 \[\deg_{\G\backslash \wt\G_{L-1,i}}(x)=\deg_{\G }(x)-\deg_{  \wt\G_{L-1,i}}(x)=\begin{cases}
     1, \ &x \ {\textrm{is an   interior boundary point of }}\ \wt\G_{L-1,i}\\
     0,  \ \ & {\rm otherwise}
 \end{cases},  \]
then by the definition \eqref{eqn:Lap-D} of the modified Dirichlet boundary condition, we have
\[\ipc{f}{ -\Delta^{\wt\G_{L-1,i},D}   f} =\ipc{f}{ -\Delta^{\wt\G_{L-1,i},N}   f}   +2\sum_{\substack{x\in \wt\G_{L-1,i} \\
 x\sim  \mathcal R}} f(x)^2, \] leading to \eqref{eqn:energy-bound1}.    
Adding the potential to \eqref{eqn:energy-bound1} and applying Lemma~\ref{lem:A3} with the subspaces $\wt \G_{L-1,i}$ and those spanned by each $e_j$, we get (dropping the contributions from the $e_j$ below)
\begin{align}\label{eqn:4.16}
      \mathcal N(E; -\Delta^{\G_L,N}+  V^{\G_L})\ge \sum_{i=1}^3\mathcal N(E; -\Delta^{\wt\G_{L-1,i},D}+V^{\wt \G_{L-1,i}}) .  
\end{align}
Similar to the upper bound, applying \eqref{eqn:XY} then implies  
\begin{align}
    \mathcal N(E; X^L)\ge \sum_{i=1}^3\mathcal N(E; X^{ L-1,i})-27 ,
\end{align}
for $X^L$ being any of $\{H^{\G_L,N},H^{\wt\G_L,D},H^{\wt\G_L}, H^{\wt\G_L,N}\}$. This completes the proof of \eqref{eqn:XL}. 
\end{proof}

The relation \eqref{eqn:XL} suggests the eigenvalue counting $\mathcal N(E; X^L)$ is almost (or very close to) a subadditive process (up to some constant shift). Applying \eqref{eqn:XL} inductively on each smaller $2^{L-j}$-triangle ($j=0,1,2,\cdots$) and then applying the ergodic theory/law of large numbers, we obtain the existence of the IDS for the Anderson model on the (infinite) Sierpinski lattice. 
\begin{theorem}\label{thm:IDS-exist-gen}
    The integrated density states $N(E)$ for the Anderson model $H=-\Delta+V$ \eqref{eqn:AM} on  the right half Sierpinski lattice  exists, and is a.s. a non-random function, which is defined by the following limit 
\begin{align}\label{eqn:IDS-exist-half}
        N(E)=\lim_{L\to\infty}\frac{1}{|\G_L|} \E\mathcal N(E; X^L)=\lim_{L\to\infty}\frac{1}{|\G_L|} \mathcal N(E; X^L),\ \  a.s. , 
    \end{align}
     where $X^L$ is any choice in  
     $\mathcal T_L=\bigcup_{\bullet\in\{S,N,D\}}\{H^{\G_L,\bullet},H^{\wt\G_L,\bullet}\}$.

     As a consequence, the integrated density states $N(E)$ of the Anderson model on  the full Sierpinski lattice,  exists as a non-random function, and can be defined by the following limit 
    \begin{align}\label{eqn:IDS-exist-full}
        N(E)=\lim_{L\to\infty}\frac{1}{|B_{L}|} \E \mathcal N(E; H^{ B_{L},\bullet})=\lim_{L\to\infty}\frac{1}{|B_{ L}|} \mathcal N(E; H^{ B_{ L},\bullet}), \ \   a.s.,  
    \end{align}
 with any boundary condition $\bullet=S,N$, or, $D$,  and where  
    $B_{L}=B(O,2^L)$    
     is the (graph metric) ball, centered at the origin $O$, with radius $2^L$.   
\end{theorem}

\begin{proof}
 \noindent {\bf Case I: the right half Sierpinski lattice  $\bigcup_{L}\G_L$.}    Let $\mathcal T_L=\bigcup_{\bullet\in\{S,N,D\}}\{H^{\G_L,\bullet},H^{\wt\G_L,\bullet}\}$ be as in the theorem.   
    Suppose $\{V(x)\}_{x\in \G}$ are i.i.d. random variables. We first show that the limit
    \begin{align}\label{eqn:expNL}
     n_E:=   \lim_{L\to\infty}\frac{1}{3^L}\E \mathcal N(E; X^L)
    \end{align}
    exists and depends only on $E$ (with the same value for any choice of $X^L\in \mathcal T_L$). 
Retain the same notations as in Lemma \ref{lem:XL-induc}. Given $X^L$, denote by $ X^{L-1,i},i=1,2,3$ the corresponding Schr\"odinger operator on one of the smaller component $2^{L-1}$-triangles, $\G_{L-1,i}$ (or $\wt \G_{L-1,i}$, respectively)  with the same boundary condition as $\{X^L\}$. All eigenvalue counting functions $\mathcal N(E; X^{L-1,i}), i=1,2,3$,   have the same expectation value since $\G_{L-1,i}$ (or $\wt \G_{L-1,i}$) are isometric to $\G_{L-1}$ (or $\wt \G_{L-1}$ respectively), and $\{V(x)\}_{x\in \G}$ are i.i.d. 
Taking the expectation in \eqref{eqn:XL} of Lemma~\ref{lem:XL-induc} gives  
    \begin{align} 
        \E\mathcal N(E; X^L)\le 3\E\mathcal N(E; X^{L-1})+30. 
    \end{align}
    Thus for fixed $E$, the limit $n_E=\lim_{L\to\infty}\frac{1}{3^L} \E\mathcal N(E; X^L)$ exists since the number sequence $a_L:=\frac{1}{3^L} \big(\E\mathcal N(E; X^L)+15\big)$ is decreasing. Clearly, $n_E$ does not depend on the choice of $X^L$ due to \eqref{eqn:XY}.

Next, we study the a.s. limit (the second equality in \eqref{eqn:IDS-exist-half}).  For any $L> \ell \ge 1$, we apply \eqref{eqn:XL} inductively 
down to the $2^\ell$-triangle size to obtain
\begin{align}
  \mathcal N(E; X^L)\le \sum_{i=1}^{3^{L-\ell}}\mathcal N(E; X^{\ell,i})+  
  15\cdot 3^{L-\ell}.
\end{align}
 Dividing both sides by $3^L$  gives
 \begin{align}\label{eqn:220}
    \frac{1}{3^L} \mathcal N(E; X^L)\le \frac{1}{3^\ell} \cdot \frac{1}{3^{L-\ell}}\sum_{i=1}^{3^{L-\ell}}\mathcal N(E; X^{\ell,i})+ 15\cdot3^{ -\ell} .  
 \end{align}
 If $X^{\ell,i}$ is any one of $\{H^{\wt\G_{\ell,i},D},H^{\wt\G_{\ell,i}}, H^{\wt\G_{\ell,i},N}\}$ where $\wt \G_{\ell,i}$ is the truncated triangle isometric to $\wt \G_\ell$, then $\{\mathcal N(E; X^{\ell,i})\}_i$ are identically distributed with the common mean $\E \mathcal N(E; X^{\ell})$.  In addition, they are all independent since $\{\wt \G_{\ell,i}\}_{i=1,2,3}$ are disjoint. Hence, by the (strong) law of large numbers, for fixed $\ell$,
 \begin{align}\label{eqn:221}
    \lim_{L\to \infty} \frac{1}{3^{L-\ell}}\sum_{i=1}^{3^{L-\ell}}\mathcal N(E; X^{\ell,i})=\E \mathcal N(E; X^{\ell}), \ \ a.s. 
 \end{align}
For fixed $\ell$, taking the limit as $L\to \infty$ in \eqref{eqn:220} thus gives 
\begin{align} 
 \limsup_{L\to \infty}   \frac{1}{3^L} \mathcal N(E; X^L)\le \frac{1}{3^\ell} \E \mathcal N(E; X^{\ell})+ 15\cdot3^{ -\ell}, \ \ a.s. 
 \end{align}
Then taking the limit as $\ell\to \infty$ and recalling the definition \eqref{eqn:expNL} of $n_E$, we obtain 
\begin{align} 
 \limsup_{L\to \infty}   \frac{1}{3^L} \mathcal N(E; X^L)\le n_E, \ \ a.s. 
 \end{align}
The same argument via the lower bound in \eqref{eqn:XL} gives 
\begin{align} 
 \liminf_{L\to \infty}   \frac{1}{3^L} \mathcal N(E; X^L)\ge n_E, \ \ a.s. 
 \end{align}
 Putting the two together, we obtain 
 \begin{align} \label{eqn:NE-limit-as} \lim_{L\to \infty}   \frac{1}{3^L} \mathcal N(E; X^L)= n_E=\lim_{L\to\infty}\frac{1}{3^L}\E \mathcal N(E; X^L), \ \ a.s. 
 \end{align}
 We proved the above limit for $X^{L}\in \{H^{\wt\G_{L},D},H^{\wt\G_{L}}, H^{\wt\G_{L},N}\}$ where we used the independence of the eigenvalue counting on disjoint triangles $\wt \G_{\ell,i}$. 
 However, due to \eqref{eqn:XY}, Eq.~\eqref{eqn:NE-limit-as} holds for $X^L=H^{\G_L,\bullet}$ as well. 

Finally, since $ |\G_L|=\frac{1}{2}(3^L+3)$,
the following limit exists 
\begin{align}\label{eqn:IDS-half-limit-pf}        N(E)=2n_E=\lim_{L\to\infty}\frac{1}{|  \G_L|} \E\mathcal N(E; X^L)=\lim_{L\to\infty}\frac{1}{|  \G_L|} \mathcal N(E; X^L),\ \  a.s., 
    \end{align}
     where $X^{L}$ is any  choice of $ \mathcal T_L$. 
 
\noindent {\bf Case II: the full   Sierpinski lattice  $\bigcup_{L}(\G_L\cup \G_L')$. }
Notice that $B_{L}=B(O,2^L)=\G_L\cup \G_L'$ and $\G_L\cap \G_L'=\{O\}$, where $\G_L'$ is the reflection of $\G_L$ with respect to the $y$-axis. By the same argument used in Lemma~\ref{lem:XL-induc}, one can obtain
\begin{align}\label{eqn:full-half}
   \mathcal N(E; H^{\wt \G_L,D})+   \mathcal N(E; H^{\wt \G_L',D}) 
    \le    \mathcal N(E; H^{B_{L},N})    \le   \mathcal N(E; H^{\G_L,N})+  \mathcal N(E; H^{\G_L',N})+1. 
\end{align}

Since $\G_L'$ is isometric to $\G_L$, \eqref{eqn:IDS-half-limit-pf}         holds for $\G_L'$ (with any boundary condition). And the resulted limit for $\G_L'$ equals the limit for $\G_L$, still denoted by $N(E)$.  Using $ |B_L|=2|\G_L|-1$ and \eqref{eqn:full-half}, we obtain 
\begin{align}\label{eqn:full-half2}
   \lim_{L\to \infty} \frac{1}{|B_L|}\mathcal N(E; H^{B_{L},N})=\frac{1}{2}N(E)+\frac{1}{2}N (E)=N(E),\  a.s. ,  
\end{align}
which defines the integrated density of states on the full Sierpinski lattice.  
The other boundary conditions on $B_L$ can be proved similarly. 
\end{proof}


 \section{Lifshitz tails for the Anderson model on the Sierpinski lattice}
 \label{sec:lif}

Throughout this section, we set $\alpha=\log3/\log 2, \beta=\log5/\log 2$. 
  Because of \eqref{eqn:full-half} and \eqref{eqn:full-half2} in the previous section, it is enough to study the tail behavior of the IDS $N(E)$ only on the right half Sierpinski lattice.   
We will prove 
\begin{theorem}\label{thm:IDS-Lif}
Let   $H_\omega=-\Delta+V_\omega$ be the Anderson model as in \eqref{eqn:AM}.   Suppose  $\{V_\omega(x)\}_{x\in \G}$ are i.i.d random variables with a (non-trivial) common distribution $P_0$, satisfying 
\begin{align}\label{eqn:V-Lif-ass}
 \inf \supp P_0=0, \ {\rm and}\    P_0([0,\eps])\ge C\eps^\kappa, 
\end{align}
for some $C,\kappa>0$ and all sufficiently small $\varepsilon>0$.
Let $N(E)$ be the IDS given by Theorem~\ref{thm:IDS-exist-lif}.  Then    
 \begin{align}\label{eqn:lif-SG}
  \lim_{E\searrow 0}   \frac{\log \big|\log  N  (E)\big|}{\log E}=-\frac{\alpha}{\beta}. 
\end{align}
\end{theorem}

In the following, we again omit the subscript $\omega$ and write $H=H_\omega$ and $V=V_\omega$. 
The proof relies on the method called Dirichlet--Neumann bracketing as reviewed in Section~\ref{sec:background}. 
 The original Dirichlet--Neumann bracketing principle refers to bounds on the spectrum obtained through additive schemes on a (disjoint) partition of the vertex set of a graph. The  Neumann and Dirichlet Laplacians (and the associated Schr\"odinger operators) are picked so that the corresponding quadratic forms give a pair of complementary bounds  on the original Hamiltonian. 
 
 Here, we continue to use Laplacians on finite triangles with the three boundary conditions defined in \eqref{eqn:Lap-simple}, \eqref{eqn:Lap-N} and \eqref{eqn:Lap-D} in the previous section. We will divide the large triangle $\G_L$ into small triangles (fundamental subdomains) of size $2^\ell$, where $2^\ell\sim E^{-1/\beta}$ is picked according to the energy level $E$. Then the quadratic form of the operator on an arbitrarily large triangle can be approximated by a sum of  quadratic forms on these small triangles, leading to the desired bound on the eigenvalue counting function by the Rayleigh--Ritz principle.

\subsection{Neumann bracketing and the Lifshitz tail upper bound }
For the upper bound, we will use the non-disjoint cover $\mathcal P$ defined in \eqref{eqn:Part1-SG} and the Neumann Laplacian on each of the small $2^\ell$-triangles, where $2^\ell\sim E^{-1/\beta}$ will be specified later. More precisely, given $L>\ell>0$, write 
\begin{align}
    \G_L=\bigcup_{j=1}^{3^{L-\ell}}\G_{\ell,j},
\end{align}
where $\G_{\ell,j}$ are all the $2^\ell$-triangles in $\G_L$.  Let $-\Delta^{\G_{\ell,j},N}$ be the Neumann Laplacian on $\G_{\ell,j}$ as in \eqref{eqn:Lap-N}. The same argument in \eqref{eqn:Lap-N-quad}, inductively on $\G_L,\G_{L-1},\cdots, \G_\ell$ (and all the triangles isometric to them), gives 
\[ 
\ipc{f}{\Delta^{\G_L,N} f}_{\ell^2(\G_L)}=\sum_{\G_{\ell,j}\subseteq \G_L  }\ipc{f}{\Delta^{\G_{\ell,j},N} f}_{\ell^2(\G_{\ell,j})}. 
\]

 Since the cover is \emph{not} disjoint, considering the overlapped onsite potential at the extreme vertices as in \eqref{eqn:Lap-N-quad-sch}, we obtain 
 \begin{align}\label{eqn:N-brack}
    \ipc{f}{ \big(-\Delta^{\G_L,N}+ V^{\G_L}\big)  f}_{\ell^2(\G_L)}   \ge \sum_{\G_{\ell,j}\subseteq \G_L  }\ipc{f}{\big(-\Delta^{\G_{\ell,j},N}+ \frac{1}{2}V^{\G_{\ell,j}}\big)  f}_{\ell^2(\G_{\ell,j})}.  
 \end{align}
As before, let $\mathcal N(E;X)=\#\{{\rm eigenvalues}\ E'\ {\rm of}\ X \ {\rm such\ that\ }\ E'\le E\}$ be the eigenvalue counting function of an operator $X$ below the energy $E$. Applying Lemma \ref{lem:NH<NH12} to \eqref{eqn:N-brack}, we obtain 
\begin{align}\label{eqn:N-sum-upper}
    \mathcal N(E; -\Delta^{\G_L,N}+ V^{\G_L}) \le    \sum_{\G_{\ell,j}\subseteq \G_L  }  \mathcal N(E;-\Delta^{\G_{\ell,j},N}+ \frac{1}{2}V^{\G_{\ell,j}}).
\end{align}
Next we divide both sides by $|\G_L|=\frac{1}{2}(3^{L+1}+3)\ge 3^L$ and take the expectation value. Since there are $3^{L-\ell}$
terms in   the   sum  which are identically distributed,  the above inequality yields 
\begin{align}\label{eqn:SG-Lif-upper}
     \frac{1}{|\G_L|} \E\mathcal N(E; H^{\G_L,N} )\le \frac{3^{L-\ell}}{3^L}  \E\mathcal N(E;-\Delta^{\G_{\ell },N}+ \frac{1}{2}V^{\G_{\ell }})\le   \P \big( E_0 \le E \big),
\end{align}
where $E_0=E_0(H^\ell)$ is the smallest eigenvalue of $H^\ell=-\Delta^{\G_{\ell },N}+ \frac{1}{2}V^{\G_{\ell }}$, and we used that $\mathcal N(  E;H^\ell)\le  |\G_\ell|  \one_{ \{E_0 \le E \} } \le 3^\ell \one_{ \{E_0 \le E \} }    $ in the last inequality.

It is enough to bound $\P \big( E_0(H^\ell)\le E \big)$  from above. This will be achieved by bounding $E_0\big( H^{ \ell}\big)$ from below.  The key ingredient is the following  Temple's  inequality.
\begin{proposition}[Temple, \cite{temple1928theory}] \label{prop:temple}
    Let $H$ be a self-adjoint operator with an isolated non-degenerate eigenvalue $E_0=\inf \sigma(H)$, and let $E_1=\inf \big(\sigma(H) \backslash\{E_0\}\big)$
.  Then for any
$\psi\in \mathcal D (H)$  (domain of $H$), which satisfies $\ipc{\psi}{H\psi}<E_1$, $\|\psi\|=1$, then the following bound holds:
\begin{align}\label{eqn:temple}
    E_0\ge \ipc{\psi}{H\psi}-\frac{\ipc{H\psi}{H\psi}-\ipc{\psi}{H\psi}^2}{E_1-\ipc{\psi}{H\psi}}. 
\end{align}
\end{proposition}
The proof of Temple's inequality  can be found in e.g. \cite{simon1985lifschitz,kirsch2007invitation,aizenman2015random}. To apply Temple's inequality to $H^\ell$, we also need the lower bound of the Neumann Laplacian eigenvalue. 
\begin{proposition}\label{prop:Neumann-ev}
    Let $E_1=E_1(-\Delta^{\G_\ell,N})$ be the first non-zero eigenvalue of the Neumann Laplacian $-\Delta^{\G_\ell,N}$. There are numerical constants $c_0=15/2,c_0'=60$ such that \begin{align}\label{eqn:Neumann-ev}
   \frac{c_0}{2^{\ell \beta }} \le      E_1\le \frac{c_0'}{2^{\ell \beta }}. 
    \end{align}
\end{proposition}
\begin{remark}
To apply Temple's inequality, we only need the lower bound of $E_1$. The upper bound is provided for completeness.  Note that $\G_\ell$ is actually the (half-sided) ball $B_R=B(O,R)$ (with respect to the gasket graph metric), centered at the origin $O=(0,0)$, with radius $R=2^\ell$. The proposition is thus equivalent to the (two-sided) asymptotic behavior $E_1(-\Delta^{B_R,N}) \sim R^{-\beta}$.
Note that the first (smallest) Dirichlet eigenvalue on a ball or a triangle of the same size has the same order asymptotic $R^{-\beta}$; see Proposition~\ref{prop:Diri-ev}. 
\end{remark}

 We will first use these two propositions to complete the proof of the Lifshitz upper bound. 
Proposition \ref{prop:Neumann-ev} will use the explicit iteration formula of the Neumann eigenvalues from \cite{teplyaev1998spectral}. The proof is  left to the end of the section.

\begin{proof}[Proof of the upper bound of Eq.~\eqref{eqn:lif-SG}]
Denote by $E_0(X),E_1(X)$ the first and second smallest eigenvalue   of   an operator $X$ respectively in the proof. 
    We consider a truncated potential
    \begin{align}\label{eqn:wtV-def}
        \wt V(x):=\min \left\{\frac{1}{2}V(x), \frac{c_0}{3}2^{-\ell \beta }\right\} ,
    \end{align}
    where $c_0=15/2$ is the constant given in \eqref{eqn:Neumann-ev}. 
    Let $\wt H^{ \ell}:= -\Delta^{\G_{\ell },N}+ \wt V^{\G_{\ell }}$, and recall $H^\ell=-\Delta^{\G_\ell,N}+\frac{1}{2}V^{\G_\ell}$. 
    Clearly,    $\wt H^{ \ell}  \le H^\ell$ by the definition of $\wt V$.   By the min-max principle \eqref{eqn:NH1<NH2}, then
    \begin{align}\label{eqn:SGE00}
      E_0\big( \wt H^{ \ell}\big)\le   E_0\big( H^{ \ell}\big)  .
    \end{align}
The rest of the work is bounding $E_0\big( \wt H^{ \ell}\big)$ from below by Temple's inequality. We will apply Proposition \ref{prop:temple} to $\wt H^{ \ell}$, with $E_0=E_0\big( \wt H^{ \ell}\big)$, $E_1=E_1\big( \wt H^{ \ell}\big)$, and 
     \[\psi(x)=\frac{1}{\sqrt{|\G_\ell|}}, \ x\in \G_\ell \]
     being the normalized constant ground state of the Neumann Laplacian $-\Delta^{\G_\ell, N}$. 
    
To proceed, we need a lower bound of $E_1\big( \wt H^{ \ell}\big)$.   Combining the min-max principle, the inequality $  \wt H^{ \ell}\ge -\Delta^{\G_\ell,N}$ and the lower bound of $E_1$ in Proposition~\ref{prop:Neumann-ev}, we obtain
\[ E_1\big( \wt H^{ \ell}\big)\ge E_1(-\Delta^{\G_\ell,N})\ge c_0\frac{1}{2^{\ell \beta }}. \]
     Recall that $\psi= |\G_\ell|^{-1/2} $ is the constant eigenfunction of $\Delta^{\G_\ell,N}$ associated with the eigenvalue $0$. Then  $\Delta^{\G_\ell,N}\psi=0$, and so 
     \begin{align}
       \ipc{\psi}{\wt H^{ \ell}\psi}=\ipc{\psi}{\wt V^{\G_\ell}\psi}= \frac{1}{|\G_\ell|}\sum_{x\in \G_\ell}\wt V(x)\le \frac{c_0}{3}2^{-\ell \beta }< E_1\big( \wt H^{ \ell}\big).
     \end{align}
     Hence, the conditions of Temple's inequality are all met. The second term on the right hand side of Temple's inequality \eqref{eqn:temple} can be bounded from above as
     \begin{align}
   \frac{\ipc{\wt H^{ \ell}\psi}{\wt H^{ \ell}\psi}-\ipc{\psi}{\wt H^{ \ell}\psi}^2}{E_1-\ipc{\psi}{\wt H^{ \ell}\psi}}  \le  &\,     \frac{\ipc{\wt H^{ \ell}\psi}{\wt H^{ \ell}\psi} }{ E_1-\ipc{\psi}{\wt H^{ \ell}\psi}} 
   =  \frac{\ipc{\wt V^{\G_\ell}\psi}{\wt V^{\G_\ell}\psi} }{ E_1-\ipc{\psi}{\wt H^{ \ell}\psi}}\nonumber \\
   \le &\, \frac{\frac{c_0}{3}2^{-\ell \beta }|\G_\ell|^{-1}\sum_{x\in \G_\ell}\wt V(x) }{ c_02^{-\ell \beta }-\frac{c_0}{3}2^{-\ell \beta }}=\frac{1}{2|\G_\ell|}\sum_{x\in \G_\ell}\wt V(x).  \label{eqn:5.14}
     \end{align}
Applying Temple's inequality \eqref{eqn:temple}, together with \eqref{eqn:SGE00} and \eqref{eqn:5.14}, thus gives 
     \begin{align}
    E_0\big(   H^{\G_\ell}\big)\ge    E_0\big( \wt H^{ \ell}\big)\ge &\,    \ipc{\psi}{\wt H^{ \ell}\psi}-\frac{\ipc{\wt H^{ \ell}\psi}{\wt H^{ \ell}\psi}-\ipc{\psi}{\wt H^{ \ell}\psi}^2}{E_1-\ipc{\psi}{\wt H^{ \ell}\psi}} \\
  \ge &\,     \frac{1}{|\G_\ell|}\sum_{x\in \G_\ell}\wt V(x)-\frac{1}{2|\G_\ell|}\sum_{x\in \G_\ell}\wt V(x)=\frac{1}{2|\G_\ell|}\sum_{x\in \G_\ell}\wt V(x). \label{eqn:Temple-app}
     \end{align}
  Note that $\{2^{\ell \beta}\wt V(x)\}_{x\in \G_\ell}$ are i.i.d. random variables with range in $[0,c_0/3]$, and with common mean 

\begin{align}
    \mu_\ell=\E\Big(\min\big\{2^{\ell \beta}V(x),\ \frac{c_0}{3} \big\}\Big)\ge \frac{c_0}{3}\P\big(2^{\ell \beta}V(x)>\frac{c_0}{3}\big)=\frac{c_0}{3}\Big[1-\P\big(V(x)\le\frac{c_0}{3\cdot 2^{\ell \beta}}\big)\Big]. 
\end{align}
Therefore, 
\begin{align}\label{eqn:inf-mu}
    \liminf_{\ell\to \infty} \mu_\ell\ge \frac{c_0}{3}[1-\P(V(x)=0)]=:\frac{c_0}{3}p_1>0, 
\end{align}
using that $p_0=1-p_1=\P(V(x)=0)<1$ since the distribution is non-trivial (the support contains more than one point). 
Then for $E>0$, let 
 \begin{align}\label{eqn:lE}
     \ell=\Big\lfloor \frac{1}{\beta\log 2}\log\Big(\frac{c_0p_1}{16}E^{-1} \Big)\Big\rfloor
 \end{align}
 so that 
 \begin{align}\label{eqn:lE2}
  2^{\ell \beta+1}E\le \frac{c_0p_1}{8}, \ \ {\rm and}\ \ 2^\ell\ge \frac{1}{2}\Big(\frac{c_0p_1}{16 }\Big)^{1/\beta} \cdot E^{-\frac{1}{\beta}}  .
 \end{align}
Combing the first inequality in  \eqref{eqn:lE2} with \eqref{eqn:Temple-app} and \eqref{eqn:SG-Lif-upper}, we obtain 
\begin{align}\label{eqn:5.21}
 \frac{1}{|\G_L|} \E\mathcal N(E; H^{\G_L,N} )\le \P \big( E_0 \le E \big)\le
   \P \Big( \frac{1}{2|\G_\ell|}\sum_{x\in \G_\ell}\wt V(x)\le E \Big)
\le \P \Big( \frac{1}{|\G_\ell|}\sum_{x\in \G_\ell}2^{\ell \beta}\wt V(x)\le \frac{c_0p_1}{8} \Big),
\end{align}
and so the problem is reduced to estimating the right-most probability from above.

Applying  the standard type of large deviation estimate (Hoeffding inequality)  to $2^{\ell \beta}\wt V(x)$,  for $E$ sufficiently small (the smallness depending only on $c_0,p_0$), we obtain 
\begin{align}\label{eqn:LDT-app}
 \P \left( \frac{1}{|\G_\ell|}\sum_{x\in \G_\ell}2^{\ell \beta}\wt V(x)\le\frac{c_0p_1}{8}\right) \le e^{-cE^{-\frac{\alpha}{\beta}} },
\end{align}
    where $c$ only depends on  $c_0,p_0,\alpha,\beta$  (in particular, is independent of $E,\ell$). The argument is quite standard and close to the proof for the $\Z^d$ case, except for the choice of the size of the fundamental domain $\ell$ in \eqref{eqn:lE} due to the specific volume control parameter and the walk dimension on the Sierpinski lattice. We sketch the proof of \eqref{eqn:LDT-app} for completeness. 
    The exponentially decaying probability estimate is provided by Hoeffding's inequality for sums of bounded i.i.d. random variables.
    \begin{proposition}[Hoeffding {\cite{hoe1963}}]\label{prop:hoeffding}
        If $\{Y_k\}_{1\le k\le K}$ are i.i.d. random variables ranging in $[0,b]$, then for any $\eps>0$, there is $c=c(\eps,b)>0$ such that 
        \begin{align}
            \P\Big(\frac{1}{K}\sum_{k=1}^KY_k-\E(Y_k)\ge \eps\Big)\le e^{-cK}. 
        \end{align}
    \end{proposition}
We have $Y_k=2^{\ell \beta}\wt V_k$ and $K=|\G_\ell|$. Using \eqref{eqn:inf-mu}, take $\ell$ sufficiently large (depending only on $c_0$ and $p_1$)  so that $\mu_\ell\ge c_0p_1/4$. Then Proposition~\ref{prop:hoeffding} gives
\begin{align}
  \P \Big( \frac{1}{|\G_\ell|}\sum_{x\in \G_\ell}2^{\ell \beta}\wt V(x)\le \frac{c_0p_1}{8} \Big) 
  \le &\,\P \Big( \frac{1}{|\G_\ell|}\sum_{x\in \G_\ell}2^{\ell \beta}\wt V(x)-\mu_\ell\le \frac{c_0p_1}{8}-\frac{c_0p_1}{4} \Big) \\
  \le &\, e^{-c|\G_\ell|},
\end{align}
    where $c>0$ only depends on $c_0$ and $p_1$. 
By the  volume lower bound $|\G_\ell|\ge 3^\ell= 2^{\ell \alpha}$, and the lower bound $2^\ell \gtrsim E^{-1/\beta}$ from \eqref{eqn:lE2}, we see 
   $|\G_\ell|\ge c'E^{-\frac{\alpha}{\beta}}$ for some constant $c'$ depending only on $c_0,p_0,\alpha,\beta$. This completes the proof of \eqref{eqn:LDT-app}. The smallness condition of $E$ is determined by the largeness requirement of $\ell$ through the relation \eqref{eqn:lE}. 
    
    Thus \eqref{eqn:5.21} becomes
     \begin{align}
   \frac{1}{|\G_L|} \E\mathcal N(E; H^{\G_L,N} ) \le    \P \big( E_0(H^{ \ell})\le E \big)  \le      \P \big( \frac{1}{2|\G_\ell|}\sum_{x\in \G_\ell}\wt V(x)\le E \big)\le ce^{-c_1E^{-\alpha/\beta}}. 
     \end{align}
     Finally, for fixed $E$, taking the limit of $L\to \infty$ and using \eqref{eqn:IDS-exist-half}, we obtain $ N(E)\le  ce^{-c_1E^{-\alpha/\beta}}$. Then taking the double log limit as $E\searrow 0$ implies the desired Lifshitz tail upper bound        \[ \limsup_{E\searrow 0}   \frac{\log \big|\log  N  (E)\big|}{\log E}\le -\frac{\alpha}{\beta}. \]
\end{proof}

It remains to complete the 
\begin{proof}[Proof of Proposition~\ref{prop:Neumann-ev}]
Let $\Delta^{\G_\ell,N}$ be the (combinatorial) Laplacian on $\ell^2(\G_\ell)$ with Neumann boundary condition, i.e., it is the subgraph Laplacian on $\G_\ell$. Denote the associated probabilistic Laplacian by $\Delta^{\G_\ell,N}_p=D  \Delta^{\G_\ell,N} $, where $D={\rm Diag}\{  \deg_{\G_\ell}(x)^{-1}  \}$ is the 
multiplication operator (diagonal matrix) by the reciprocal of the vertex degree (so all entries are either $1/2$ or $1/4$).  All the eigenvalues of $\Delta^{\G_\ell,N}_p$ can be explicitly determined by the decimation method as described in \cite{shima1991eigenvalue,teplyaev1998spectral}. We will use the following formulation in \cite[Proposition 3.12]{teplyaev1998spectral}: For $\ell\ge 1$, the eigenvalues of $\Delta^{\G_\ell,N}_p$ are given by 
\begin{align}\label{eqn:Tep-eigen}
    \sigma(\Delta^{\G_\ell,N}_p)=\left\{-\frac{3}{2}\right\}\cup \left(\bigcup_{m=0}^{\ell-1} R_{-m}\left\{0,-\frac{3}{4}\right\}\right),
\end{align}
where $R(z)=z(4z+5)$, and $R_{-m}A$ is the preimage of a set $A\subseteq \R$ under the $m$-th composition power of $R$. For any $x\in \R$, 
its preimage under $R$ is 
\[R_{-1}\{x\}=\Big\{\ \frac{-5-\sqrt{25+16x}}{8},\ \frac{-5+\sqrt{25+16x}}{8}\ \Big\}.  \]
Denote the larger root  (at least for $x\ge-25/16$)  by 
\begin{align}\label{eqn:deci-f}
f(x)=\frac{-5+\sqrt{25+16x}}{8}.
\end{align}
By \eqref{eqn:Tep-eigen} and the monotonicity of $f$, the largest eigenvalue of $\Delta^{\G_\ell,N}_p$ is $0$, and the second largest  eigenvalue of $\Delta^{\G_\ell,N}_p$ is the $(\ell-1)$-th iteration of $-3/4$ under $f$, i.e., $f^{\circ(\ell-1)}(-3/4)=f\circ f \circ \cdots \circ f(-3/4)$. 
Computation of the series expansion of $f$ with the Taylor remainder theorem shows that for $-1\le x\le0$,  
\begin{align}
   \frac{1}{5}x(1-x) \le f(x)\le \frac{1}{5}x, 
\end{align}
which implies that for $n\ge 0$ and $-1\le x\le 0$
\begin{align}\label{eqn:f-it}
       \frac{4}{5^n}x\le f^{\circ n}(x)\le \frac{1}{5^n}x.
\end{align}
The upper bound is immediate, the lower bound of \eqref{eqn:f-it} can be proved by a direct induction and the constant is not optimal. We include the computation in  Appendix~\ref{sec:f-it}.

By \eqref{eqn:f-it}, we obtain 
\begin{align}
 -15\frac{1}{2^{\beta \ell}}  =- \frac{3}{4}\frac{4}{5^{\ell-1}}\le f^{\circ(\ell-1)}(-3/4)\le -\frac{3}{4}\frac{1}{5^{\ell-1}}=-\frac{15}{4}\frac{1}{2^{\beta \ell}}, 
\end{align}
 in which we used $5^\ell=2^{\ell \log 5/\log 2}=2^{\beta\ell}$. 
In other words, the first (smallest) eigenvalue of $-\Delta^{\G_\ell,N}_p$ is $E_0(-\Delta^{\G_\ell,N}_p)=0$, and the second (the first non-zero) eigenvalue of $-\Delta^{\G_\ell,N}_p$ is $E_1(-\Delta^{\G_\ell,N}_p)=- f^{\circ(\ell-1)}(-3/4)$ satisfying 
\begin{align}\label{eqn:E1-lower-p}
 \frac{15}{4}\frac{1}{2^{\beta \ell}}\le   E_1(-\Delta^{\G_\ell,N}_p) \le  15\frac{1}{2^{\beta \ell}}. 
\end{align}

Note that both  $D $ and $-\Delta^{\G_\ell,N}$ are positive semidefinite. In addition, the largest eigenvalue of $D $ is $1/2$. Then by the  majorization theory of eigenvalues \eqref{eqn:EjAB}, one has 
\begin{align}
 \frac{1}{2}E_1\big( -\Delta^{\G_\ell,N} \big)\ge    E_1\Big(D (-\Delta^{\G_\ell,N})\Big)=E_1\Big( -\Delta_p^{\G_\ell,N} \Big) ,
\end{align}
which, together with \eqref{eqn:E1-lower-p}, implies that 
\[ E_1(-\Delta^{\G_\ell,N})  \ge 2E_1(-\Delta^{\G_\ell,N}_p) \ge \frac{15}{2}\frac{1}{2^{\beta \ell}}.  \]
Similarly, using that the smallest eigenvalue of $D $ is $1/4$ and applying \eqref{eqn:EjAB} again, we obtain
\[ E_1(-\Delta^{\G_\ell,N})  \le 4E_1(-\Delta^{\G_\ell,N}_p) \le 60\frac{1}{2^{\beta \ell}}.  \]
\end{proof}


\subsection{(Modified) Dirichlet bracketing and the Lifshitz tail lower bound}

For the lower bound of \eqref{eqn:lif-SG}, we use the disjoint partition $\wt{\mathcal P}$ \eqref{eqn:Part2-SG} and the modified Dirichlet Laplacian \eqref{eqn:Lap-D}. Similar to the Neumann bracketing, given $L>\ell>0$, write 
\begin{align}
    \G_L=\Bigg(\bigcup_{j=1}^{3^{L-\ell}}\wt \G_{\ell,j}\Bigg)\cup \mathcal R, 
\end{align}
but where $\wt \G_{\ell,j}$ are all the truncated (and hence disjoint) $2^\ell$-triangles associated with $\G_{\ell,j}\subseteq \G_L$ and $\mathcal R$ is the collection of all the extreme vertices of all $\G_{L,j}\subseteq \G_L$; see Figure~\ref{fig:Vpart2}. \begin{figure}
    \centering
    \includegraphics[width=0.5\linewidth]{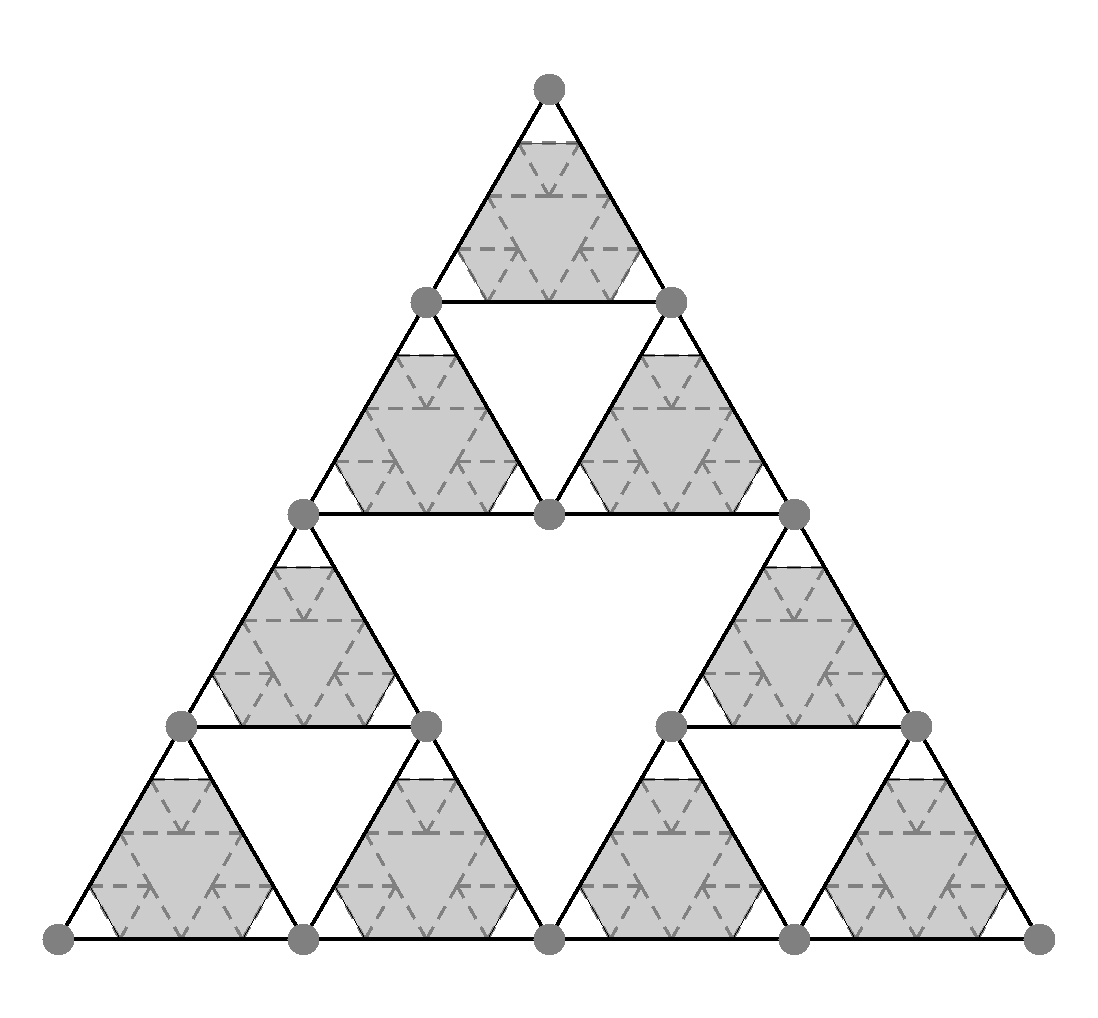}
    \caption{The $2^4$-triangle $\G_4$ (the big triangle) consists of 9 $2^2$-triangle $\{\G_{2,i}\}_{i=1}^9$. The $2^2$-triangles have totally 15 extreme vertices (the filled circles) $\mathcal R=\{e_j\}_{j=1}^{15}$.  
    Remove $\mathcal R$ from $\G_4$, then the truncated triangles $\wt \G_{2,i}$ (the shadowed ones) are disjoint, and $(\cup_{i=1}^9\wt \G_{2,i})\cup \mathcal R$ form a disjoint partition of $\G_4$.}
    \label{fig:Vpart2}
\end{figure}

Let  $-\Delta^{\wt \G_{\ell,j},D}$ be the   (modified) Dirichlet Laplacian on $\wt \G_{\ell,j}$ as in \eqref{eqn:Lap-D}.   Using the same argument as in \eqref{eqn:energy-bound1}  to bound the removed edge energy between   $\G_{\ell,j}$ and $\mathcal R$, we obtain  
\begin{align}
  \ipc{f}{-\Delta^{\G_L,N} f}_{\ell^2(\G_L)}\le \sum_{\wt \G_{\ell,j}\subseteq \G_L  }\ipc{f}{-\Delta^{\wt \G_{\ell,j},D} f}_{\ell^2(\wt \G_{\ell,j})}+8\sum_{e_j\in \mathcal R}f(e_j)^2,   
\end{align}
which implies 
\begin{align}
  \ipc{f}{(-\Delta^{\G_L,N}+V^{\G_L}) f}_{\ell^2(\G_L)}\le \sum_{\wt \G_{\ell,j}\subseteq \G_L  }\ipc{f}{(-\Delta^{\wt \G_{\ell,j},D}+V^{\wt \G_{\ell,j}}) f}_{\ell^2(\G_{\ell,j})}+\sum_{e_j\in \mathcal R}(8+V_{e_j})f(e_j)^2.   
\end{align}

Then by Lemma~\ref{lem:A3}, dropping terms on the right hand side as in \eqref{eqn:4.16},
\begin{align}
    \mathcal N( E; H^{\G_L,N} ) \ge  \sum_{\wt \G_{\ell,j}\subseteq \G_L  }  \mathcal N(E;    H^{\wt \G_{\ell,j},D} ).
\end{align}

Taking the expectation both sides, using the fact that all $\wt \G_{\ell,j}$ are isometric to $\G_\ell$ and $\{V(x)\}$ are i.i.d., we obtain 
\begin{align}\label{eqn:D-direct-sum-lower}
  \E \mathcal N(E; H^{\G_L,N}) \ge 3^{L-\ell} \E \mathcal N( E; H^{\wt \G_{\ell},D}).
\end{align}
 Let $E_0=E_0(H^{\wt \G_{\ell},D})$ be the ground state energy of $H^{\wt \G_{\ell},D}$. Given $E>0$, if $E_0(H^{\wt \G_{\ell},D})\le E$, then $\mathcal N(E; H^{\wt \G_{\ell},D})$ is at least one. Hence, 
\begin{align} \label{eqn:D-N-P-lower}
     \E \mathcal N(E; H^{\wt \G_{\ell},D})\ge \P\Big(E_0(H^{\wt \G_{\ell},D})\le E\Big).
\end{align}
It is enough to bound $\P\Big(E_0(H^{\wt \G_{\ell},D})\le E \Big)$ from below, or equivalently, to bound the ground state energy $E_0(H^{\wt \G_{\ell},D})$ from above. 
In order to make use of estimates for the ground state energy of the Laplacian with {simple} boundary conditions (Proposition~\ref{prop:Diri-ev}), we consider slightly smaller truncated triangles that avoid the boundary vertices.
For $\ell$ sufficiently large, let $\wt T \subsetneq \wt \G_\ell$ be a truncated $2^{\ell_1}$-triangle with side length $\ell_1=\ell-2$, and  located away from the (interior) boundary vertices $\{o_i\}_{i=1}^6$  of  $\wt \G_\ell$ (see Figure \ref{fig:wtB}).
\begin{figure}
    \centering
\includegraphics[width=.5\textwidth]{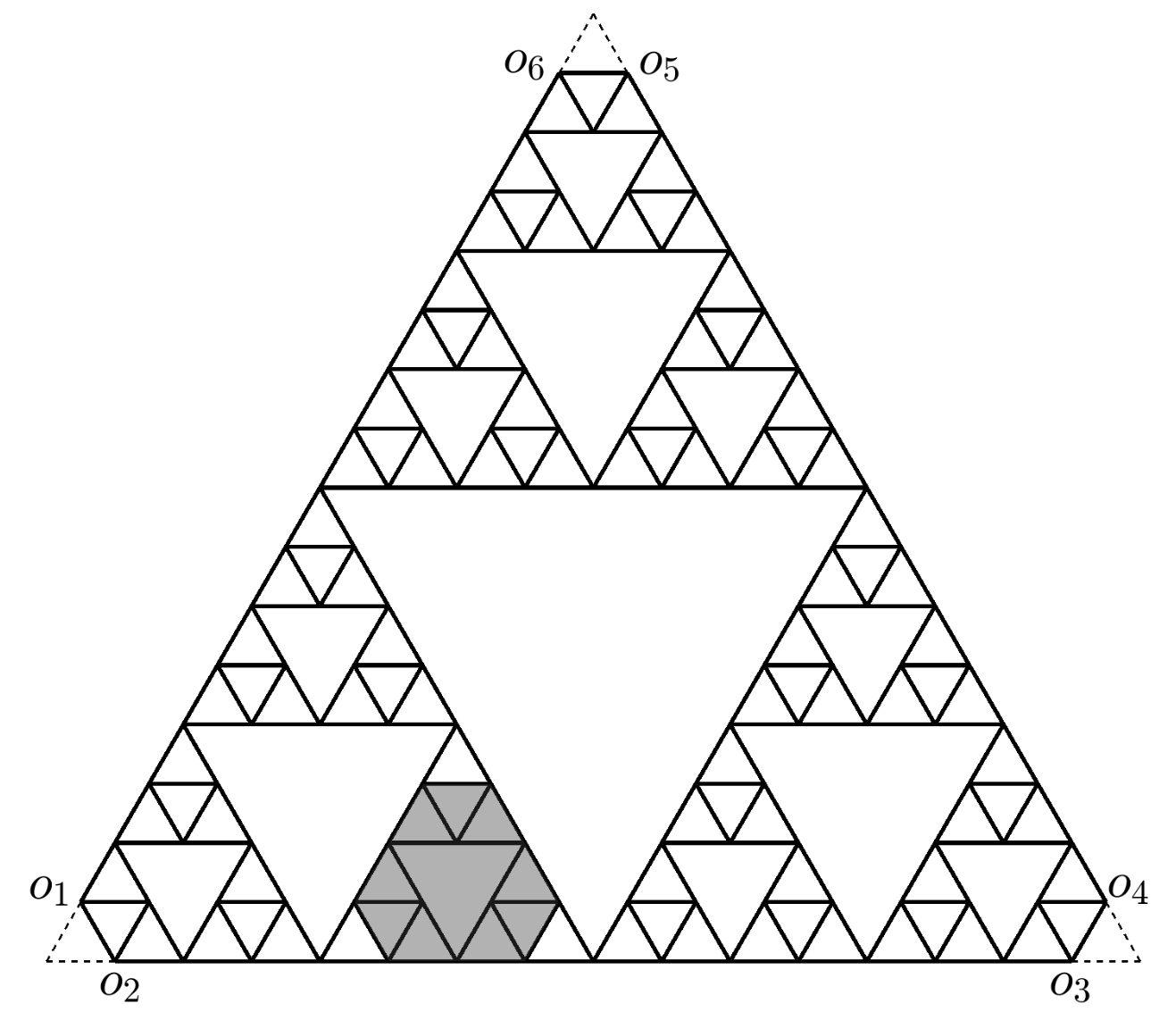}
    \caption{The entire figure is a truncated $2^4$-triangle $\wt \G_4$. The shaded region is a smaller truncated $2^2$-triangle $\wt T\subsetneq \wt \G_4$, located near the midpoint of the bottom edge so that $\wt \T$ is strictly away from the 6 interior boundary vertices of $\wt \G_4$. }
    \label{fig:wtB}
\end{figure}
Let $\phi_0$ be the ground state of the Laplacian $-\Delta^{\wt T}$ on $\wt T$, with the simple boundary condition as in \eqref{eqn:Lap-simple}. Extend $\phi_0$ to $\ell^2(\wt \G_\ell)$ by setting $\phi_0(x)=0$ on $\wt \G_\ell\backslash \wt T$.
Then $-\Delta^{\wt \G_\ell,D}\phi_0=-\Delta^{\wt T}\phi_0$ since $\wt T$ is located away from the interior boundaries of $\wt \G_\ell$. 
By the min-max principle, then
\begin{align}
   E_0(H^{\wt \G_{\ell},D})=\inf_{\varphi\neq 0}\frac{\ipc{\phi}{H^{\wt \G_{\ell},D}\phi}}{\ipc{\phi}{\phi}} \le & \,\inf_{\varphi\neq 0}\frac{\ipc{\phi}{-  \Delta^{\wt \G_\ell,D}\phi}}{\ipc{\phi}{\phi}}+\max_{\wt \G_\ell}V(x) \nonumber\\
     \le & \,\frac{\ipc{\phi_0}{-  \Delta^{\wt \G_\ell,D}\phi_0}}{\ipc{\phi_0}{\phi_0}}+\max_{\wt \G_\ell}V(x)
      = \frac{\ipc{\phi_0}{-\Delta^{\wt T }\phi_0}}{\ipc{\phi_0}{\phi_0}}+\max_{\wt \G_\ell}V(x) \nonumber\\
       = &\, E_0\big(-\Delta^{\wt T }\big)+\max_{\wt \G_\ell}V(x). \label{eqn:E0-1}
\end{align}
We need the following upper bound of the simple (zero Dirichlet) Laplacian eigenvalue on truncated triangles.
\begin{proposition}\label{prop:Diri-ev}
Let $\wt \G_\ell \subseteq \G$ be a truncated $2^\ell$-triangle. 
Let $E_0(-\Delta^{\wt \G_\ell})$ be the ground state energy (smallest eigenvalue) of the  Laplacian with simple boundary conditions $-\Delta^{\wt \G_\ell}$. There are numerical constants $c_0=40,c'_0=10$, such that for any $\ell\in \N$,     \begin{align}\label{eqn:Diri-E0-upper-tri}
  \frac{c'_0}{2^{\ell \beta}}   \le    E_0(-\Delta^{\wt \G_\ell})\le \frac{c_0}{2^{\ell \beta}} . 
    \end{align}

\end{proposition} 
    This is the analogue of Proposition~\ref{prop:Neumann-ev} for the simple (zero Dirichlet) Laplacian. One proof is again based on the recursive expression for the eigenvalue obtained by the  decimation method in \cite{shima1991eigenvalue,teplyaev1998spectral}. By \cite[\S 6]{teplyaev1998spectral}, the ground state eigenvalue is given by
    \begin{align}
         E_0( \Delta^{\wt \G_\ell})=4f^{\circ(\ell-1)}\Big(-\frac{1}{2}\Big),
    \end{align}
    where $f(x)$ is the same function as in \eqref{eqn:deci-f}. Using again the iteration estimate of $f^{\circ(n)}(x)\sim x/5^n$ in \eqref{eqn:f-it}, one obtains \eqref{eqn:Diri-E0-upper-tri} with $c_0=40$ and $c_0'=10$. We omit the details here. 
    
\begin{remark}    
    By the min-max principle, the same estimate holds for the Dirichlet Laplacian on a (non-truncated) $2^\ell$-triangle $\G_\ell$. 
    More generally, the asymptotic behavior of the first Dirichlet eigenvalue on a graph ball, $ E_0(-\Delta^{B(x,r)}) \approx r^{-\beta}$, always holds on graphs satisfying the Heat Kernel Bound $\mathrm{HK}(\alpha,\beta)$, using the fact that $E_0$ is always proportional to the reciprocal of the exist time from balls; see e.g. \cite[Corollary 2]{shou2024}.     
\end{remark}

Applying the upper bound of  \eqref{eqn:Diri-E0-upper-tri} to  \eqref{eqn:E0-1} with $\wt T=\wt \G_{\ell_1},\ell_1=\ell-2$, we arrive at 
\begin{align}
   E_0(H^{\wt \G_{\ell},D})\le \frac{c_1}{2^{\ell \beta}}+\max_{\wt \G_\ell}V(x). 
\end{align}
Let 
\begin{align*}
    \ell=\Big\lceil \frac{1}{\beta\log 2}\log\big(2c_1 E^{-1}\big) \Big\rceil,
\end{align*}
so that 
\begin{align}\label{eqn:lE-upper}
    \frac{c_1}{2^{\ell \beta }}\le \frac{1}{2}E,\ \ {\rm and}\ \ \ \  2^\ell\le 2 (2c_1)^{\frac{1}{\beta}}\cdot E^{-\frac{1}{\beta}}.
\end{align}
Then $E_0(H^{\wt \G_{\ell},D})\le E/2+\max_{\wt \G_\ell}V(x)  $, which implies  for sufficiently small $E$,
\begin{align}
  \P\Big(E_0(H^{\wt \G_{\ell},D})\le E\Big)\ge \P\Big(\max_{x\in \wt \G_\ell}V(x)\le E/2\Big) \ge&\;  \P\Big({\textrm{For all} } \ x\in\wt \G_\ell, V(x)\le E/2\Big)\nonumber \\
 \ge &\; C(E/2)^{\kappa |\wt \G_\ell|} \\
  \ge &\; Ce^{c_2 (\log(E/2)) E^{-\frac{\alpha}{\beta}}}, \label{eqn:PE0-lower}
\end{align}
where we used the probability distribution assumption $\P\big(V(x)\le E\big)\ge CE^{\kappa}$  from \eqref{eqn:V-Lif-ass}, and the upper bound of 
\[|\wt \G_\ell|\le  3^{\ell+1}=3\times 2^{\ell \alpha}\lesssim E^{-\alpha/\beta}  \]
 from \eqref{eqn:lE-upper} (since $\log(E/2)$ is negative for small $E$). Putting \eqref{eqn:PE0-lower} together with \eqref{eqn:D-direct-sum-lower} and \eqref{eqn:D-N-P-lower}, we obtain 
\begin{align}
    \frac{1}{|\G_L|} \E\mathcal N(E; H^{\G_L,N})\ge \frac{3^{L-\ell}}{|\G_L|} \P\Big(E_0(H^{\wt \G_{\ell},D})\le E\Big)\ge   c_3E^{-\frac{\alpha}{\beta}}e^{c_2 (\log(E/2)) E^{-\frac{\alpha}{\beta}}}.
\end{align}
 Similarly  as for the upper bound, for fixed $E$, taking the limit as $L\to \infty$ and again using \eqref{eqn:IDS-exist-half}, we obtain
        $ N(E)\ge  c_3E^{-\frac{\alpha}{\beta}}e^{c_2 (\log(E/2)) E^{-\frac{\alpha}{\beta}}}$. Then taking the double log limit as $E\searrow 0$ implies the desired Lifshitz tail lower bound
        \[ \liminf_{E\searrow 0}   \frac{\log \big|\log  N  (E)\big|}{\log E}\ge -\frac{\alpha}{\beta}. \]

\appendix 
\section{Eigenvalue counting comparison}\label{sec:min-max}
The min-max principle by E. Fischer \cite{fischer1905quadratische}
and R. Courant \cite{courant1920eigenwerte} for self-adjoint operators (bounded from below) is a useful tool to count the eigenvalues  below the essential spectrum of the operator.  In this appendix, we briefly summarize some of the consequences of the min-max principle for   self-adjoint linear operators $H$ (actually real symmetric matrices) on a finite dimensional Hilbert space $\mathcal H=\C^K$. In this case, $H$ has eigenvalues $\lambda_1\le \cdots\le \lambda_K$ and the eigenvalue counting function is denoted as  $ \mathcal N(E;H)=\#\{i:\lambda_i\le E\}$. Clearly, 
$ \mathcal N(E;H)={\rm dim\ span}\{\phi^i\in \mathcal H:\ H\phi^i=\lambda_i\phi^i,\lambda_i\le E\}={\rm dim}\{f\in \mathcal H:\ipc{f}{Hf}\le E\ipc{f}{f}\}$.  As a consequence, if 
     $\ipc{f}{Hf}> E\ipc{f}{f}$,  for all $f$ in a subspace $S\subseteq \mathcal H,$
 then \begin{align}\label{eqn:codim}
      \mathcal N(E;H)\le {\rm codim}S. 
 \end{align}
 and if 
     $\ipc{f}{Hf}\le E\ipc{f}{f}$,  for all $f\in S,$
 then \begin{align}\label{eqn:dim}
      \mathcal N(E;H)\ge {\rm dim}S. 
 \end{align}
 
 For the matrix case,  the min-max principle reads (see e.g. \cite[{Theorem 4.13}]{aizenman2015random})
\begin{align}
    \lambda_n=\min_{\phi^1,\cdots,\phi^n \in \mathcal H}\max\big\{\ipc{\phi}{H\phi}:\ \ \phi\in {\rm span}(\phi^1,\cdots,\phi^n) , \ \|\phi\|=1\big\}.
\end{align}
Hence, if   $H_i,i=1,2$ are two self-adjoint   operators on   $\mathcal H$ satisfying $H_1\ge H_2$ in the sense of positive operators on $\mathcal H$, then 
    \begin{align}\label{eqn:NH1<NH2}
        \mathcal N(E;H_1) \le \mathcal N(E;H_2). 
    \end{align}
 We would like to compare the eigenvalue counting functions between matrices up to some rank one perturbations or orthogonal projections onto a smaller subspace. 

\begin{lemma} \label{lem:A1}
    Let $H$ be  a real symmetric matrix  on $\mathcal H=\C^K$. Let $H_1$ be the orthogonal projection of $H$ acting on a linear subspace $\mathcal H_1\cong \C^{L}$.  Then 
\begin{align}\label{eqn:NH-cauchy-inter}
    \mathcal N(E;H_1)\le  \mathcal N(E;H)\le  \mathcal N(E;H_1)+{\rm codim}H_1.
\end{align}
Now suppose $H_2=H+\sum_{j=1}^mc_je_{i_j}e^T_{i_j}$ is a rank $m$ diagonal perturbation of $H$. Then 
\begin{align}\label{eqn:NH-diag-pert}
      |\mathcal N(E;H)-\mathcal N(E;H_2)|\le   m.
\end{align}
\end{lemma}
\begin{proof}
    Suppose $H$ has eigenvalues $\lambda_1\le \cdots\le \lambda_K$ and $H_1$ has eigenvalues $\mu_1\le \cdots\le \mu_{L}$. By the Cauchy interlacing theorem, 
\begin{align}
    \lambda_i\le \mu_i\le \lambda_{i+K-L}, i=1,\cdots,L.
\end{align}
which implies \eqref{eqn:NH-cauchy-inter}. 

For the second part, note that $H$ and $H_2$ have the same orthogonal projection $\wt H$, obtained by deleting the rows and columns from $H$ for $i_1,\cdots,i_m$, onto the subspace $\wt {\mathcal H}=\{e_{i_1},\ldots,e_{i_m}\}^\perp$ with ${\rm codim}\wt {\mathcal H}=m$. Then \eqref{eqn:NH-cauchy-inter}  implies 
\begin{align}
  \mathcal N(E; H_2 )-m\le \mathcal N(E;\wt H )\le   \mathcal N(E;H)\le  \mathcal N(E;\wt H )+m\le \mathcal N(E;  H_2 )+m, 
\end{align}
which is 
\eqref{eqn:NH-diag-pert}.  
\end{proof}

\begin{lemma}\label{lem:NH<NH12}
Let $H$ be a self-adjoint matrix on $\mathcal{H}\cong\C^n$, and let $\{v_1,\ldots,v_n\}$ be an orthonormal basis for $\mathcal{H}$. For $i=1,\ldots,k$, set $\mathcal{H}_i=\operatorname{span}(\{v_m:m\in I_i\})$ for sets $I_i\subseteq\{1,\ldots,n\}$, and let $H_i$ be a self-adjoint operator on $\mathcal{H}_i$ for each $i=1,\ldots,k$.
Suppose $\bigcup_{i=1}^kI_i=\{1,\ldots,n\}$, and that
\begin{align}\label{eqn:H>H12}
\ipc{f}{Hf}_{\mathcal H}\ge \sum_{i=1}^k\ipc{P_if}{H_iP_if}_{\mathcal H_i}, \quad\text{for all }f\in\mathcal H,
\end{align}
where $P_i$ is the orthogonal projection onto $\mathcal H_i$.
 Then \begin{align}\label{eqn:NH<H12}
      \mathcal N(E;H)\le  \sum_{i=1}^k\mathcal N(E;H_i),
 \end{align}
 where $\mathcal N(E;H_i)$ is the eigenvalue counting function for the  restriction  of $H_i$ on $\mathcal H_i$.  
\end{lemma}
\begin{proof}
    For any $i$, let $S_i=\{\phi\in \mathcal H_i\subseteq\mathcal H: \ipc{\phi}{H_i\phi}\le E\ipc{\phi}{\phi}\}$. 
    Let $  S=S_1+S_2+\cdots+ S_k\subseteq \mathcal H$.   
    Suppose 
     $f \in S^\perp\subseteq\cap_i  S_i^{\perp} $. 
    For each $i$, write $f=P_if+ P_i^{\perp}f$, where $P_i^\perp$ is the orthogonal projection onto $\mathcal{H}_i^\perp$. If $P_if\neq 0$, then $f \in    S_i^{\perp}$ implies $P_if \in S_i^{\perp}$.
    By the definition of $S_i$,  then
  $  \ipc{P_if}{H_iP_if}_{\mathcal H_i}\ge E\ipc{P_if}{P_if}$. Note if $P_if= 0$, the same inequality holds trivially. 
 Then \eqref{eqn:H>H12} implies that 
 \begin{align}
    \ipc{f}{Hf}_{\mathcal H}\ge   \sum_{i=1}^k E\ipc{P_if}{P_if} \ge E\ipc{f}{f}, 
  \end{align}
  where in the last inequality we used the definition of the $\mathcal{H}_i$ and that $\bigcup_{i=1}^kI_i=\{1,\ldots,n\}$.
 Therefore, by \eqref{eqn:codim}, \begin{align}
    \mathcal N(E;H)\le {\rm codim}S^\perp={\rm dim} S \le \sum_{i=1}^k {\rm dim}  S_i . 
 \end{align}
 Recalling that $\mathcal N(E;H_i)$ is taken on the Hilbert space $\mathcal{H}_i$, we see that 
 ${\rm dim}S_i= \mathcal N(E;H_i)$ by definition. 
\end{proof}

\begin{lemma}\label{lem:A3} 
Let $H$ be a real symmetric matrix on $\mathcal H$. Let $\mathcal H_i\subseteq \mathcal H$, for $i=1,\cdots,k$ with $k\ge2$, be subspaces, and let $H_i$ be a 
self-adjoint linear  operator  on  $\mathcal H_i$ for each $i=1, \cdots,k$.
Suppose 
that $\mathcal H_i$ and $\mathcal H_j$ are orthogonal for $i\neq j$, and that
\begin{align}\label{eqn:H<H12}
    \ipc{f}{Hf}_{\mathcal H}\le \sum_{i=1}^k\ipc{P_if}{H_iP_if}_{\mathcal H_i}, \quad\text{for all }f\in\mathcal H,
\end{align}
 where $P_i$ is the orthogonal projection onto $\mathcal H_i$.
 Then 
 \begin{align}\label{eqn:NH>H12}
      \mathcal N(E;H)\ge  \sum_{i=1}^k\mathcal N(E;H_i),
 \end{align}
  where $\mathcal N(E;H_i)$ is the eigenvalue counting function for the  restriction  of $H_i$ on $\mathcal H_i$.
 In particular, $\mathcal N(E;H)\ge  \mathcal N(E;H_i)$ for any $i$. 
\end{lemma}
\begin{proof}
    Let $S_i=\{\phi\in \mathcal H_i: \ipc{\phi}{H_i\phi}\le E\ipc{\phi}{\phi}\}$. 
    Let $  S=\oplus_i S_i\subseteq\mathcal{H}$ be the direct sum of $S_i$ (which is well-defined since $S_i\subseteq \mathcal H_i$ are orthogonal).  
    Suppose $f \in S$, so that $P_if\in S_i$, which implies  $  \ipc{P_if}{H_iP_if}_{\mathcal H_i}\le E\ipc{P_if}{P_if}$ by the definition of $S_i$.  
 Then \eqref{eqn:H<H12} implies that 
 \begin{align}
    \ipc{f}{Hf}_{\mathcal H}\le  \sum_{i=1}^k  E\ipc{P_if}{P_if}\le E\ipc{f}{f}, 
  \end{align}
  where in the last inequality we used that the $ {S}_i$ are orthogonal so that $\sum_{i=1}^k P_i$ is the orthogonal projection onto $S$.
 Therefore, by \eqref{eqn:dim}, \begin{align}
    \mathcal N(E;H)\ge   {\rm dim} S =\sum_{i=1}^k {\rm dim}  S_i=\sum_{i=1}^k \mathcal N(E;H_i). 
 \end{align}
\end{proof}

\section{Asymptotic behavior of the iteration \eqref{eqn:f-it}}\label{sec:f-it}
Suppose that for $-1\le x\le0$,
\begin{align}\label{eqn:f-lower}
  0\ge   f(x)\ge \frac{1}{5}x(1-x).
\end{align}
We show that for $n\ge 1$, the $n$-th iteration of $f$ satisfies, for $-1\le x \le 0$,
\begin{align}\label{eqn:fn-lower}
 0\ge    f^{\circ n}(x)\ge  \frac{1}{5^n}x \big(1- S_n  x\big),\ \ \text{where } S_n=\sum_{m=0}^{n-1}\frac{(m+1)^2}{5^m}. 
\end{align}
Since $S_1=1$, \eqref{eqn:fn-lower} holds for $n=1$. Now suppose that \eqref{eqn:fn-lower} holds for some $n \ge 1$. 
Note that \eqref{eqn:f-lower} implies that $0\ge f^{\circ m}(x)\ge -1$ for $-1\le x\le0$ for any $m\ge1$.
Then 
\begin{align*}
    f^{\circ ({n+1})}(x)=   
    f\big(f^{\circ n}(x)\big)  
    \ge &\,  \frac{1}{5}f^{\circ n}(x)(1-f^{\circ n}(x)) \\
 \ge &\,  \frac{1}{5}\frac{1}{5^n}x \big(1- S_n  x\big)\Big[1-\frac{1}{5^n}x \big(1- S_n  x\big) \Big]  \\
 =& \, \frac{1}{5^{n+1}}x  \Big[1- S_n  x-\frac{1}{5^n}x \big(1- S_n  x\big)^2   \Big]. \numberthis
\end{align*}
Using the very loose bound$0<S_n\le n$ and $-1\le x\le 0$, one has 
$1- S_n  x\le 1+n$.

Hence,
\begin{align}
  f^{\circ ({n+1})}(x)\ge  \frac{1}{5^{n+1}}x  \Big[1- S_n  x-\frac{1}{5^n}x  (1+n )^2   \Big]=   \frac{1}{5^{n+1}}x (1- S_{n+1}  x  ),
\end{align}
which completes the induction. Direct computation gives $S_n\le S_\infty=75/32$ for all $n$. Therefore,  for $-1\le x \le 0$,
\begin{align}
   f^{\circ n}(x)\ge  \frac{1}{5^n}x \big(1- S_\infty  x\big)\ge \frac{1}{5^n}x \left(1+\frac{75}{32}\right)\ge    \frac{4}{5^n}x ,
\end{align}
which is the lower bound in \eqref{eqn:f-it}. 
\section{Ordered eigenvalues of the product of two semidefinite   matrices}
\begin{lemma}
    Denote by $E_0(X)\le E_1(X)\le \cdots E_{n-1}(X)$ the ordered eigenvalues (if they are all real) of a $n\times n$  matrix $X$.  Suppose $A,B$ are two $n\times n$ positive semidefinite Hermitian matrices. Then for   $j=0,\cdots,n-1$, 
    \begin{align}\label{eqn:EjAB}
 E_{0}(A)E_j(B) \le      E_j(AB)\le E_{n-1}(A)E_j(B).
    \end{align}
\end{lemma}
This is some well-known result of the theory majorization of eigenvalues. A more general version can be found in e.g. \cite[p.340, H.1.d. Theorem (Lidski\v{i}, 1950)]{marsh2011}. We sketch the proof of the special case \eqref{eqn:EjAB} here for the reader's convenience. 
\begin{proof}
    Since $B$ is positive semidefinite Hermitian, then $B^{1/2}$ is well-defined, and is also positive semidefinite Hermitian. Rewrite $AB=(AB^{1/2})\cdot B^{1/2}$ which is similar to $B^{1/2}\cdot(AB^{1/2})$. Hence, the ordered eigenvalues (counting multiplicity) of $AB$ are the same as $B^{1/2} AB^{1/2} $, i.e., 
    \[E_{j}(AB)=E_j(B^{1/2} AB^{1/2}),\ \  j=0,\cdots,n-1.  \]
    On the other hand, denote by $E_{\max}=E_{n-1}(A)$ the largest eigenvalue of $A$. Then $E_{\max}-A\ge 0$  implies  $B^{1/2}E_{\max}B^{1/2}\ge B^{1/2}AB^{1/2}$ (both inequalities being in the positive semidefinite sense). Then by the min-max principle, one has for   $j=0,\cdots,n-1$,  
    \[E_j(B^{1/2} AB^{1/2})\le E_j(B^{1/2}E_{\max}B^{1/2})=E_{\max}\cdot E_j(B^{1/2}  B^{1/2})=E_{\max}\cdot E_j(B ), \]
    which is the upper bound. The lower bound can be proved exactly in the   same way using  $B^{1/2}E_{0}(A)B^{1/2}\le B^{1/2}AB^{1/2}$. 
\end{proof}

\noindent\textbf{Acknowledgments.}
\phantomsection
\addcontentsline{toc}{section}{Acknowledgments}
The authors would like to thank Jeffrey Schenker for many useful discussions about the Anderson model on the Sierpinski gasket, in particular, for suggesting studying the spectrum  using the Weyl criterion. S.Z. was supported by the NSF grant DMS-2418611.


\bibliographystyle{abbrv}
\bibliography{SG.bib}

\begin{thebibliography}{10}

\bibitem{aizenman2015random}
M.~Aizenman and S.~Warzel.
\newblock {\em Random operators}, volume 168.
\newblock American Mathematical Soc., 2015.

\bibitem{alex1984}
S.~Alexander.
\newblock Some properties of the spectrum of the {S}ierpi\'nski gasket in a magnetic field.
\newblock {\em Phys. Rev. B (3)}, 29(10):5504--5508, 1984.

\bibitem{balsam2023density}
H.~Balsam, K.~Kaleta, M.~Olszewski, and K.~Pietruska-Pa\l~uba.
\newblock Density of states for the {A}nderson model on nested fractals.
\newblock {\em Anal. Math. Phys.}, 14(2):Paper No. 23, 56, 2024.

\bibitem{barlow2017random}
M.~T. Barlow.
\newblock {\em Random walks and heat kernels on graphs}, volume 438.
\newblock Cambridge University Press, 2017.

\bibitem{barlow1988}
M.~T. Barlow and E.~A. Perkins.
\newblock Brownian motion on the {S}ierpi\'nski gasket.
\newblock {\em Probab. Theory Related Fields}, 79(4):543--623, 1988.

\bibitem{belli1982stab}
J.~Bellissard.
\newblock Stability and instability in quantum mechanics.
\newblock In {\em Trends and developments in the eighties ({B}ielefeld, 1982/1983)}, pages 1--106. World Sci. Publishing, Singapore, 1985.

\bibitem{belli1988ren}
J.~Bellissard.
\newblock Renormalization group analysis and quasicrystals.
\newblock In {\em Ideas and methods in quantum and statistical physics ({O}slo, 1988)}, pages 118--148. Cambridge Univ. Press, Cambridge, 1992.

\bibitem{courant1920eigenwerte}
R.~Courant.
\newblock {\"U}ber die eigenwerte bei den differentialgleichungen der mathematischen physik.
\newblock {\em Mathematische Zeitschrift}, 7(1-4):1--57, 1920.

\bibitem{dalrymple1999fractal}
K.~Dalrymple, R.~S. Strichartz, and J.~P. Vinson.
\newblock Fractal differential equations on the {S}ierpinski gasket.
\newblock {\em J. Fourier Anal. Appl.}, 5(2-3):203--284, 1999.

\bibitem{dobrushin1993lecture}
R.~L. Dobrushin, S.~Kusuoka, and S.~Kusuoka.
\newblock Lecture on diffusion processes on nested fractals.
\newblock {\em Statistical mechanics and fractals}, pages 39--98, 1993.

\bibitem{fal1990}
K.~Falconer.
\newblock {\em Fractal geometry}.
\newblock John Wiley \& Sons, Ltd., Chichester, 1990.
\newblock Mathematical foundations and applications.

\bibitem{fischer1905quadratische}
E.~Fischer.
\newblock {\"U}ber quadratische {F}ormen mit reellen {K}oeffizienten.
\newblock {\em Monatshefte f{\"u}r Mathematik und Physik}, 16:234--249, 1905.

\bibitem{fuku1988}
M.~Fukushima.
\newblock Dirichlet forms, diffusion processes and spectral dimensions for nested fractals.
\newblock In {\em Ideas and methods in mathematical analysis, stochastics, and applications ({O}slo, 1988)}, pages 151--161. Cambridge Univ. Press, Cambridge, 1992.

\bibitem{fuku1992}
M.~Fukushima and T.~Shima.
\newblock On a spectral analysis for the {S}ierpi\'nski gasket.
\newblock {\em Potential Anal.}, 1(1):1--35, 1992.

\bibitem{goldstein1987random}
S.~Goldstein.
\newblock Random walks and diffusions on fractals.
\newblock {\em Percolation theory and ergodic theory of infinite particle systems}, pages 121--129, 1987.

\bibitem{harris1973rigorous}
A.~B. Harris.
\newblock Rigorous bound on the integrated density of states of a three-dimensional random alloy.
\newblock {\em Physical Review B}, 8(8):3661, 1973.

\bibitem{hoe1963}
W.~Hoeffding.
\newblock Probability inequalities for sums of bounded random variables.
\newblock {\em J. Amer. Statist. Assoc.}, 58:13--30, 1963.

\bibitem{kig1989}
J.~Kigami.
\newblock A harmonic calculus on the {S}ierpi\'nski spaces.
\newblock {\em Japan J. Appl. Math.}, 6(2):259--290, 1989.

\bibitem{kirsch2007invitation}
W.~Kirsch.
\newblock An invitation to random {S}chr{\"o}dinger operators.
\newblock {\em arXiv:0709.3707}, 2007.

\bibitem{kirsch1982cmp}
W.~Kirsch and F.~Martinelli.
\newblock On the spectrum of {S}chr\"odinger operators with a random potential.
\newblock {\em Comm. Math. Phys.}, 85(3):329--350, 1982.

\bibitem{kirsch1983large}
W.~Kirsch and F.~Martinelli.
\newblock Large deviations and {L}ifshitz singularity of the integrated density of states of random hamiltonians.
\newblock {\em Communications in Mathematical Physics}, 89:27--40, 1983.

\bibitem{kosi2017}
A.~Kosior and K.~Sacha.
\newblock Localization in random fractal lattices.
\newblock {\em Physical Review B}, 95(10):104206, 2017.

\bibitem{kunz1980}
H.~Kunz and B.~Souillard.
\newblock Sur le spectre des op\'erateurs aux diff\'erences finies al\'eatoires.
\newblock {\em Comm. Math. Phys.}, 78(2):201--246, 1980/81.

\bibitem{kusuoka1987diffusion}
S.~Kusuoka.
\newblock A diffusion process on a fractal.
\newblock {\em Probabilistic methods in mathematical physics, Katata/Kyoto, 1985}, 1987.

\bibitem{lapidus1994}
M.~L. Lapidus.
\newblock Analysis on fractals, {L}aplacians on self-similar sets, noncommutative geometry and spectral dimensions.
\newblock {\em Topol. Methods Nonlinear Anal.}, 4(1):137--195, 1994.

\bibitem{lif1965}
I.~M. Lifshitz.
\newblock Energy spectrum structure and quantum states of disordered condensed systems.
\newblock {\em Soviet Physics Uspekhi}, 7:549--573, 1965.

\bibitem{lind1990}
T.~Lindstr{\o}m.
\newblock Brownian motion on nested fractals.
\newblock {\em Mem. Amer. Math. Soc.}, 83(420):iv+128, 1990.

\bibitem{malo1993}
L.~Malozemov.
\newblock Spectral theory of the differential {L}aplacian on the modified {K}och curve.
\newblock In {\em Geometry of the spectrum ({S}eattle, {WA}, 1993)}, volume 173 of {\em Contemp. Math.}, pages 193--224. Amer. Math. Soc., Providence, RI, 1994.

\bibitem{ma95}
L.~Malozemov and A.~Teplyaev.
\newblock Pure point spectrum of the {L}aplacians on fractal graphs.
\newblock {\em J. Funct. Anal.}, 129(2):390--405, 1995.

\bibitem{manna2024}
S.~Manna, B.~a. Jaworowski, and A.~E.~B. Nielsen.
\newblock Many-body localization on finite generation fractal lattices.
\newblock {\em J. Stat. Mech. Theory Exp.}, (5):Paper No. 053301, 15, 2024.

\bibitem{marsh2011}
A.~W. Marshall, I.~Olkin, and B.~C. Arnold.
\newblock {\em Inequalities: theory of majorization and its applications}.
\newblock Springer Series in Statistics. Springer, New York, second edition, 2011.

\bibitem{pastur1980}
L.~A. Pastur.
\newblock Spectral properties of disordered systems in the one-body approximation.
\newblock {\em Comm. Math. Phys.}, 75(2):179--196, 1980.

\bibitem{pietruska1991lifschitz}
K.~Pietruska-Paluba.
\newblock The {L}ifschitz singularity for the density of states on the {S}ierpinski gasket.
\newblock {\em Probability theory and related fields}, 89(1):1--33, 1991.

\bibitem{ram1984}
R.~Rammal.
\newblock Spectrum of harmonic excitations on fractals.
\newblock {\em J. Physique}, 45(2):191--206, 1984.

\bibitem{reed1981functional}
M.~Reed and B.~Simon.
\newblock {\em I: Functional analysis}, volume~1.
\newblock Academic press, 1981.

\bibitem{schreib1996}
M.~Schreiber and H.~Grussbach.
\newblock Dimensionality dependence of the metal-insulator transition in the {A}nderson model of localization.
\newblock {\em Phys. Rev. Lett.}, 76:1687--1690, Mar 1996.

\bibitem{shima1991lifschitz}
T.~Shima.
\newblock Lifschitz tails for random {S}chrodinger operators on nested fractals.
\newblock {\em Osaka J. Indust. Appl. Math.}, 8:127--141, 1991.

\bibitem{shima1991eigenvalue}
T.~Shima.
\newblock On eigenvalue problems for the random walks on the {S}ierpi\'nski pre-gaskets.
\newblock {\em Japan J. Indust. Appl. Math.}, 8(1):127--141, 1991.

\bibitem{shou2024}
L.~Shou, W.~Wang, and S.~Zhang.
\newblock Landscape approximation of the ground state eigenvalue for graphs and random hopping models.
\newblock {\em J. Funct. Anal.}, 286(7):Paper No. 110339, 57, 2024.

\bibitem{simon1985lifschitz}
B.~Simon.
\newblock Lifschitz tails for the {A}nderson model.
\newblock {\em Journal of statistical physics}, 38:65--76, 1985.

\bibitem{stic2016}
D.~Sticlet and A.~Akhmerov.
\newblock Attractive critical point from weak antilocalization on fractals.
\newblock {\em Physical Review B}, 94(16):161115, 2016.

\bibitem{stri1999}
R.~S. Strichartz.
\newblock Isoperimetric estimates on {S}ierpinski gasket type fractals.
\newblock {\em Trans. Amer. Math. Soc.}, 351(5):1705--1752, 1999.

\bibitem{temple1928theory}
G.~Temple.
\newblock The theory of {R}ayleigh's principle as applied to continuous systems.
\newblock {\em Proceedings of the Royal Society of London. Series A, Containing Papers of a Mathematical and Physical Character}, 119(782):276--293, 1928.

\bibitem{teplyaev1998spectral}
A.~Teplyaev.
\newblock Spectral analysis on infinite {S}ierpi\'nski gaskets.
\newblock {\em J. Funct. Anal.}, 159(2):537--567, 1998.

\bibitem{teschl2014mathematical}
G.~Teschl.
\newblock {\em Mathematical methods in quantum mechanics}, volume 157.
\newblock American Mathematical Soc., 2014.

\end{thebibliography}

{
  \bigskip
  \vskip 0.08in \noindent --------------------------------------

\footnotesize
\medskip
L.~Shou, {Joint Quantum Institute, Department of Physics, University of Maryland, College Park, MD 20742, USA}\par\nopagebreak
    \textit{E-mail address}:  \href{mailto:lshou@umd.edu}{lshou@umd.edu}

\vskip 0.4cm

  W. ~Wang, {LSEC, Institute of Computational Mathematics and Scientific/Engineering Computing, Academy of Mathematics and Systems Science, Chinese Academy of Sciences, Beijing 100190, China}\par\nopagebreak
  \textit{E-mail address}: \href{mailto:ww@lsec.cc.ac.cn}{ww@lsec.cc.ac.cn}
  
\vskip 0.4cm

S.~Zhang, {Department of Mathematics and Statistics, University of Massachusetts Lowell, 
Southwick Hall, 
11 University Ave.
Lowell, MA 01854
 }\par\nopagebreak
  \textit{E-mail address}: \href{mailto:shiwen\_zhang@uml.edu}{shiwen\_zhang@uml.edu}
}

\end{document}